\documentclass[hidelinks]{article}

\usepackage{arxiv}

\usepackage[utf8]{inputenc} 
\usepackage[T1]{fontenc}    
\usepackage{hyperref}       
\usepackage{url}            
\usepackage{booktabs}       
\usepackage{amsfonts}       
\usepackage{nicefrac}       
\usepackage{microtype}      
\usepackage{apacite}

\usepackage{bm} 
\usepackage{amsmath} 
\usepackage{indentfirst} 
\usepackage[colorinlistoftodos, textwidth = 20mm]{todonotes} 
\usepackage{siunitx} 
\usepackage{enumitem} 
\usepackage{amsthm}
\usepackage{mathtools}
\usepackage{multirow}
\usepackage{placeins}
\usepackage{graphicx} 

\newcommand{\E}{\mathbb{E}}
\newcommand{\Var}{\text{Var}}

\newcommand{\Prob}{\mathbb{P}}

\newcommand{\Sum}{\sum_{i = 1}^n}
\newcommand{\bbm}{\begin{bmatrix}}
\newcommand{\ebm}{\end{bmatrix}}
\newcommand{\pder}[2]{\dfrac{\partial#1}{\partial#2}}

\newcommand{\Es}[1]{\mathbb{E}_{\bm{\theta}^{(t)}}\left[#1\middle|\bm{u}_{Oi}\right]}

\newtheorem{lemma}{Lemma}

\title{A Hybrid EM Algorithm for Linear Two-Way Interactions with Missing Data}

\date{} 					

\author{
  Dale S.~Kim\\
  Department of Statistics\\
  University of California, Los Angeles\\
  Los Angeles, CA 90095 \\
  \texttt{dalekim25@ucla.edu} \\
}

\renewcommand{\undertitle}{}

\begin{document}
\maketitle

\begin{abstract}
We study an EM algorithm for estimating product-term regression models with missing data.
The study of such problems in the likelihood tradition has thus far been restricted to an EM algorithm method using full numerical integration.
However, under most missing data patterns, we show that this problem can be solved analytically, and numerical approximations are only needed under specific conditions.
Thus we propose a hybrid EM algorithm, which uses analytic solutions when available and approximate solutions only when needed.
The theoretical framework of our algorithm is described herein, along with two numerical experiments using both simulated and real data.
We show that our algorithm confers higher accuracy to the estimation process, relative to the existing full numerical integration method.
We conclude with a discussion of applications, extensions, and topics of further research.
\end{abstract}


\section{Introduction}

We consider the problem of missing data in regression models with product-term predictors.
In the psychological sciences, product-term regression models are widely used to test hypotheses pertaining to interactions \cite{Aiken1991},  moderation \cite{Baron1986}, and/or conditional processes \cite{Hayes2018}.
For example, these hypotheses may refer to the difference of an effect between two groups, the dependence of an outcome-predictor relationship on other variables, or the effect of two simultaneous symptoms above and beyond their constituent effects.
A considerable amount of methodological research has been dedicated to interpreting these models \cite{Dawson2014, McCabe2018, Preacher2006}, attesting to their importance and popularity.

However, the estimation of product-term regression models is complicated by the issue of missing data, which is particularly prevalent for data involving human subjects.
It can arise from subject dropout, item non-response, logistical errors, or even serve as designed aspect of data collection \cite{Graham2009, Raghunathan2004}.
If the mechanism of missing data meets certain conditions, this problem (or design) can be accommodated and consistent estimates can be obtained.
In particular, we consider the situation under which the missing data mechanism is called \textit{ignorable} \cite{Schafer1997}.
Colloquially speaking, it means that the probability of an observation being missing does not depend on its own would-be realized value, and that the data value generating mechanism is distinct from the missingness generating mechanism.

Previous literature on this problem can be broadly cast into two categories: correctly specified and misspecified characterizations of the joint distribution of the data.
For product term-regression models, incorrect specification generally occurs for the product terms.
The most common type of misspecification is naively assuming the product terms are jointly Gaussian along with their constituent factors \cite{VonHippel2009}.
This has also been called the ``just another variable'' approach \cite{Seaman2012} as it treats the product term simply as another Gaussian random variable.
However, this introduces a contradiction of distributional assumptions, as a product of Gaussians random variables cannot itself be Gaussian.
While some studies have shown that there may be some conditions under which this method is reasonable \cite{Enders2014, Seaman2012}, it is not guaranteed to provide unbiased estimates in general \cite{Bartlett2015, Ludtke2019, Zhang2017}.

To address these issues, correctly specified methods have been developed.
This is typically accomplished by a factorizing the joint distribution into a product with conditional distributions.
In this way, it can be easier to correctly specify the constituent factored distributions, rather than the original joint distribution itself (e.g., \citeNP{Ibrahim1990}).
Hence, this technique can ensure compatibility between the substantive model of interest, and the overall joint distribution of the data \cite{Liu2014}.
It has also been called factored regression modeling \cite{Ludtke2019}, substantive model compatibility \cite{Bartlett2015}, and model-based handling \cite{Enders2020}.

The application of this technique however, has been largely focused on multiple imputation methods with Markov Chain Monte Carlo (MCMC) under a Bayesian paradigm \cite{Kim2015, Ludtke2020, Zhang2017}.
While MCMC methods are widely applicable, convergence assessment is notoriously difficult \cite{Brooks1998, Cowles1996, Gelman1996}, and it is known for its slow and inefficient uses of Monte Carlo samples \cite{Gelman1992, Mossel2006}.
Further, improper applications of MCMC can result in unreliable and/or biased inferences \cite{Cowles1999, Flegal2008}.
As such, maintaining the integrity of the procedure requires careful scrutiny and design.

\begin{figure}[t]
\centering
\includegraphics[width=\textwidth, keepaspectratio]{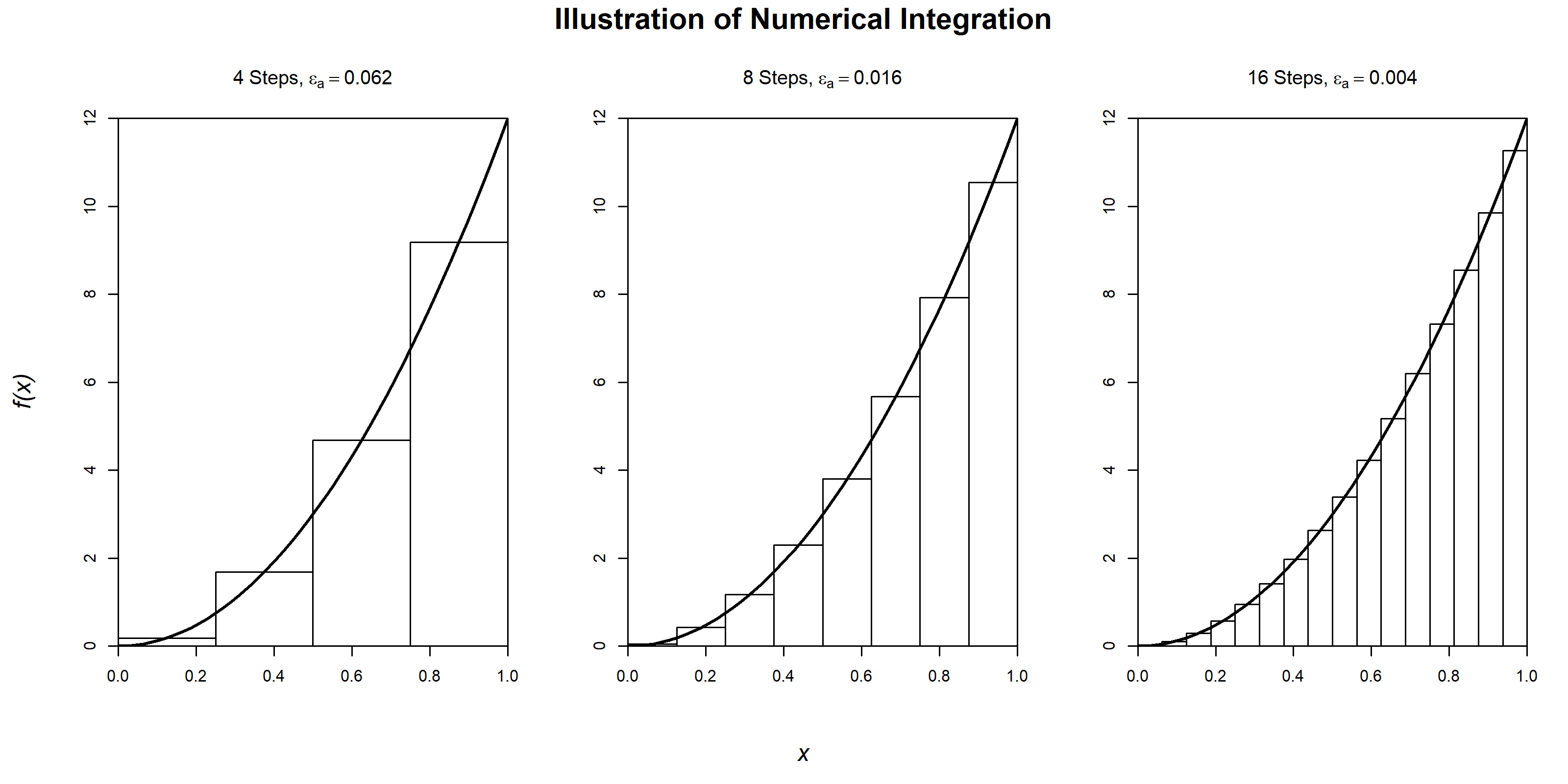} \\
\caption{Illustration of Riemann midpoint numerical integration. The error of approximation is denoted as $\varepsilon_a$.}
\label{fig:ni_ex}
\end{figure}

To avoid such issues, we focus on correctly specified likelihood techniques, for which research has been scant.
Currently, only an EM algorithm using full numerical integration has been proposed by \citeA{Ludtke2019}.
While their method is flexible and handles a variety of non-linear models, numerical integration is known to suffer in accuracy and computational complexity as the number of dimensions increase \cite{Hinrichs2014, Simonovits2003}.
Hence, the feasibility of this method is in question even when the number of variables is moderate.
Indeed to date, this method has only been tested under very optimistic conditions.
Specifically, only the simplest version of the product-term model (two predictors with one product) has been simulated, with only a single variable exhibiting missingness.

To illustrate this problem, consider the numerical integration example displayed in Figure \ref{fig:ni_ex}.
We numerically integrate an example function: $f(x) = 12 x^2$ on the interval $[0, 1]$ using the Riemann midpoint method.
Supposing it is prohibitive to integrate $f(x)$ analytically, we may instead integrate an approximation to $f(x)$ whose integration is much easier.
The Riemann midpoint method approximates $f(x)$ with a step function, whose integral is easily taken by summing the rectangles underneath each step.
As illustrated by this one-dimensional example, the error of approximation ($\varepsilon_a$) reduces rather quickly as the number of interpolating steps increase.

A different story arises when we study the multi-dimensional case.
Suppose we consider the multi-dimensional generalization of $f(x)$, by independently adding dimensions (i.e., $f(x_1, \dots, x_p) = \sum_{j = 1}^p 12x_j^2$).
The number of steps required to achieve the same amount of accuracy increases exponentially.
In Figure \ref{f_niError} (left) we show a contour plot of the error of approximation plotted against the number of steps.
Each contour shows that exponentially more steps are required in order to achieve any given level of accuracy.
For example, suppose we wanted to achieve a numerical approximation that is within 0.01 of the true value.
We can see in Figure \ref{f_niError} (right) that as low as the 4-dimensional case, we require about two hundred thousand steps, and the 5-dimensional case requires almost 6.5 million.

\begin{figure}[t]
\centering
\includegraphics[width=\textwidth, keepaspectratio]{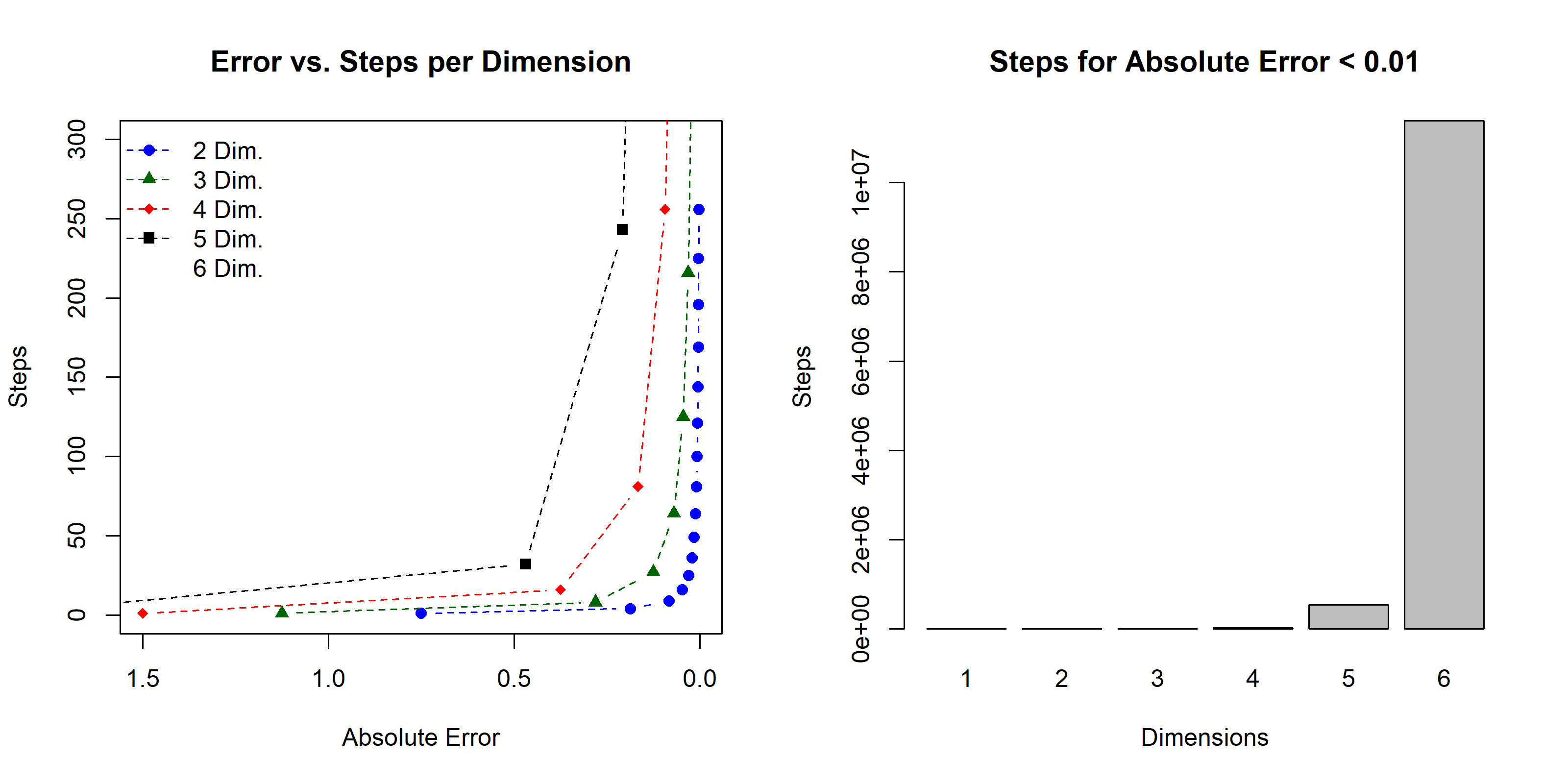} \\
\caption{Demonstration of numerical integration computational complexity and approximation error.}
\label{f_niError}
\end{figure}

Needless to say, numerical integration should not be undertaken unless absolutely necessary, and this is the premise of our proposed research.
We propose a hybrid EM algorithm that obviates much of the required calculations done by numerical integration.
Specifically, we will show that numerical integration is not necessary for most missing data patterns, and demonstrate how to use analytic solutions in their place.
These exact solutions will yield faster and more accurate estimates relative to their approximate counterparts.
Therefore, this research has two main goals: (1) to develop the theoretical motivation of the hybrid EM algorithm and demonstrate its feasibility, and (2) to empirically study the accuracy improvements conferred by the hybrid EM algorithm in practical data scenarios.

\subsection{Model and Notation}

Let $\bm{X} \sim \mathcal{N}_p(\bm{\mu}, \bm{\Sigma})$ denote a $p \times 1$ random vector of predictor variables.
Then formulate a linear product-term model for a random scalar outcome variable $Y$ as follows:
\begin{equation} \label{regressionModel}
Y = \bm{d}(\bm{X})^T \bm{\beta} + \epsilon,
\end{equation}
where $\epsilon \sim \mathcal{N}(0, \sigma^2_{\epsilon})$ is a scalar random variable of error terms, $\bm{\beta}$ is a $d \times 1$ vector of regression coefficients, and $\bm{d}(\bm{X})$ is an $d \times 1$ vector-valued design function as follows:
\begin{equation}
\bm{d}(\bm{X}) = \begin{bmatrix} 1 & \bm{X}^T & \overrightarrow{X_j X_k}^T \end{bmatrix}^T,
\end{equation}
where $\overrightarrow{f(\bm{X})}$ denotes the vector of all unique permutations of $f(\bm{X})$ over the specified indices.
In this case, $\overrightarrow{X_j X_k}$ is a vector whose elements are comprised of all unique permutations of $X_j X_k$ for all $j, k \in \{1, \dots, p\}$
Hence, $\bm{d}(\bm{X})$ is a vector that augments $\bm{X}$ with a regression intercept and the product terms.

We note that there are two implicit assumptions with this model for the purposes of generality.
First, we assume that the vector $\bm{X}$ contains the substantive model predictors as well as any desired auxiliary variables.
Second, it is assumed the substantive model contains all possible two-way products among the variables in $\bm{X}$.
To accommodate the fact that some auxiliary variables or product terms may not be desired in the substantive model, their $\beta$ coefficient need only be constrained to zero.

\subsection{Missing Data Assumptions}

To recast the data from a predictor-outcome distinction to a missing-observed distinction, we use the following notation.
Denote the an augmented data vector as $\bm{U} = \begin{bmatrix} Y & \bm{X}^T\end{bmatrix}^T$, which can be reordered as $\begin{bmatrix} \bm{U}_O^T & \bm{U}_M^T \end{bmatrix}^T$, where $O$ is the index set of observed variables and $M$ is the index set of missing variables.
Further, $\Prob(\bm{U})$ is parameterized generally by a vector $\bm{\theta}$, for which we write $\Prob_{\bm{\theta}}(\bm{U})$.
Then, let $\bm{R} \in \{0 ,1\}^{p + 1}$ be a binary random vector, which indicates whether the elements of $\bm{U}$ are observed, and has a probability distribution parameterized by the vector $\bm{\zeta}$.

We assume that the elements of $\bm{\theta}$ and $\bm{\zeta}$ are \textit{distinct}, or that the joint space of $\bm{\theta}$ and $\bm{\zeta}$ is simply their Cartesian product $\bm{\theta} \times \bm{\zeta}$.
Further, we assume the data are \textit{missing at random} (MAR; \citeNP{Rubin1976}):
\begin{equation}
\Prob_{\bm{\zeta}}(\bm{R} | \bm{U}) = \Prob_{\bm{\zeta}}(\bm{R} | \bm{U}_O).
\end{equation}
Taking MAR in tandem with the distinctness of $\bm{\theta}$ and $\bm{\zeta}$, we say that the missing data mechanism is \textit{ignorable} \cite{Schafer1997}.

\subsection{The EM Algorithm for Missing Data} \label{sec:em_alg}

The EM algorithm is a two-step iterative procedure for obtaining parameter estimates for models with missing data \cite{Dempster1977}.
The steps are as follows:

\noindent \textbf{\textit{E}-Step.} For any iteration $t$, define a $Q$-function given an intial parameter start value $\bm{\theta}^{(0)}$:
\begin{equation}
\begin{aligned}
Q_{\bm{\theta}^{(t)}}(\bm{\theta}) &= \E_{\bm{\theta}^{(t)}} \left[ \log \Prob_{\bm{\theta}}(\bm{U}) | \bm{U}_O \right] \\
&= \int_{\bm{u}_M} \log \Prob_{\bm{\theta}}(\bm{U}) \mathbb{P}_{\bm{\theta}^{(t)}}(\bm{u}^{}_M | \bm{u}^{}_O) \, d\bm{u}^{}_M.
\end{aligned}
\end{equation}

\noindent \textbf{\textit{M}-Step. } Maximize the $Q$-function with respect to $\bm{\theta}$ and set the result as $\bm{\theta}^{(t+1)}$:
\begin{equation}
\bm{\theta}^{(t + 1)} = \underset{\bm{\theta}}{\text{argmax}} \, Q_{\bm{\theta}^{(t)}}(\bm{\theta}),
\end{equation}
where integrating with respect to a vector is shorthand for multiple integration with respect to all elements of the vector (i.e., $\int_{\bm{z}} f(\bm{z}) \, d\bm{z} = \int_{z_1} \cdots \int_{z_p} f(\bm{z}) \, dz_p \cdots dz_1$, for $\bm{z} \in \mathbb{R}^p$).
Hence, this is an iterative procedure that maximizes the expectation of the complete data log-likelihood, given the observed data.
It is known to converge to a local maximum of the likelihood function under very general conditions \cite{Wu1983}.

\subsection{Application to Product-Term Regression Models} \label{sec:em_prod}

For practical uses, the main task of applying the EM algorithm is setting up the $Q$-function.
We do so for product-term regression models by characterizing the joint model of the data as follows:
\begin{equation}
\mathbb{P}(\bm{U}) = \mathbb{P}(Y | \bm{X})\mathbb{P}(\bm{X}),
\end{equation}
where,
\begin{equation}
\begin{aligned}
\mathbb{P}(\bm{X}) &\sim \mathcal{N}_p(\bm{\mu}, \bm{\Sigma}) \\
\mathbb{P}(Y | \bm{X}) &\sim \mathcal{N}(\bm{d}(\bm{x})^T \bm{\beta}, \sigma^2_{\epsilon}).
\end{aligned}
\end{equation}Since $\mathbb{P}(\bm{U})$ factorizes into two Gaussian distributions, it can be written in exponential family form:
\begin{equation}
\mathbb{P}(\bm{U}) = \exp\left[ \bm{\eta}(\bm{\theta})^T \bm{T}(\bm{U}) - A(\bm{\theta}) \right],
\end{equation}
which yields a $Q$-function of:
\begin{equation}
\begin{aligned}
Q_{\bm{\theta}^{(t)}}(\bm{\theta}) &= \mathbb{E}_{\bm{\theta}^{(t)}} \left[ \log \mathbb{P}_{\bm{\theta}}(\bm{U}) | \bm{U}_O \right] \\
&= \mathbb{E}_{\bm{\theta}^{(t)}} \left[ \bm{\eta}(\bm{\theta})^T \bm{T}(\bm{U}) - A(\bm{\theta}) | \bm{U}_O \right] \\
&= \bm{\eta}(\bm{\theta})^T \mathbb{E}_{\bm{\theta}^{(t)}} \left[ \bm{T}(\bm{U}) | \bm{U}_O \right] - A(\bm{\theta}).
\end{aligned}
\end{equation}
where $\bm{\theta}$ is the vector which contains the unique elements of $\{\bm{\beta}, \sigma^2, \bm{\mu}, \bm{\Sigma}\}$, $\bm{\eta}(\bm{\theta})$ is the vector of canonical parameters, and $A(\bm{\theta})$ is the log-partition function.
Hence, constructing the $Q$-function amounts to deriving $\mathbb{E} \left[ \bm{T}(\bm{U}) | \bm{U}_O \right]$ per missing data pattern.
It can be shown that $\bm{T}(\bm{U})$ is:
\begin{equation} \label{SS}
\bm{T}(\bm{U}) = \left[ \begin{smallmatrix}
Y & Y \overrightarrow{X^{}_j}^T & Y^2 & Y \overrightarrow{X^{}_j X^{}_k}^T & \overrightarrow{X^{}_j}^T & \overrightarrow{X^{}_j X^{}_k}^T & \overrightarrow{X^2_j}^T & \overrightarrow{X^{}_j X^{}_k X^{}_l}^T & \overrightarrow{X^2_j X^{}_k}^T & \overrightarrow{X^{}_j X^{}_k X^{}_l X^{}_m}^T & \overrightarrow{X^2_j X^{}_k X^{}_l}^T & \overrightarrow{X^2_j X^2_k}^T
\end{smallmatrix} \right]^T,
\end{equation}
recalling that $\overrightarrow{f(\bm{X})}$ denotes the vector of all unique permutations of $f(\bm{X})$ over the specified indices.
For example, $\overrightarrow{X_j X_k X_l X_m}$ is the vector comprised of all unique permutations of $X_j X_k X_l X_m$, for all $j, k, l, m \in \{1, \dots p\}$.
The derivation of $\bm{T}(\bm{U})$ has been relegated to Appendix \ref{sec_ss}.
Note that these sufficient statistics only apply to product-term regression models.
Different polynomial designs will imply a different $\bm{T}(\bm{U})$ vector.
We also derive the maximizers of the $Q$-function in Appendix~\ref{sec:q_max}.

\section{Missing Data Patterns} \label{sec:mdp}

The theoretical motivation of this research is the derivation of analytic $Q$-functions under as many missing data patterns as possible.
The form of $\bm{T}(\bm{U})$ may appear complex and the possible missing data patterns for $\mathbb{E}\left[ \bm{T}(\bm{U}) | \bm{U}_O \right]$ are combinatorially large.
However, using an appropriate taxonomy, solutions for general classes of missing data patterns can be obtained and applied easily.
The only types of missing data patterns (MDP) that need to be considered are as follows:
\begin{itemize}
 \item \textbf{MDP 1:} $Y$ is missing and $\bm{X}$ has any missingness pattern.
 \item \textbf{MDP 2:} $Y$ is observed and $\bm{X}$ is patterned such that no product terms are fully missing.
 \item \textbf{MDP 3:} $Y$ is observed and $\bm{X}$ is patterned such that one or more product terms are fully missing.
\end{itemize}
We will provide the methods of calculating $\mathbb{E}\left[ \bm{T}(\bm{U}) | \bm{U}_O \right]$ under each of these patterns.
Specifically, we will show that analytic solutions exist for MDP 1 and MDP 2, and computational methods are only necessary for MDP 3.

\subsection{$Y$ is Missing, \boldmath{$X$} has Any Missing Data Pattern} \label{sec:mdp1}

MDP 1 is concerned with the case when $Y$ is missing, and $\bm{X}$ can take on any missingness pattern.
We will show that all elements of $\mathbb{E}\left[ \bm{T}(\bm{U}) | \bm{U}_O \right]$ under this pattern can be calculated by known functions of $\bm{\theta}$.
Hence, the $Q$-function for this MDP can always be constructed analytically.

First, we consider the sufficient statistics that are solely a function of $\bm{X}$ and do not have a $Y$ term.
In Equation \ref{SS} these are the latter 8 (of 12) entries of $\bm{T}(\bm{U})$.
Note that these entries are all products of the elements of $\bm{X}$ (e.g., $X_j X_k X_l X_m$ or $X^2_j X^2_k$).
Further, $Y$ is missing in this MDP, so we have $\bm{U}_O = \bm{X}_O$.
Thus, under MDP 1, we can more generally express the elements of $\mathbb{E}\left[ \bm{T}(\bm{U}) | \bm{U}_O \right]$ that only depend on $\bm{X}$ as:
\begin{equation}
\label{e_condProd}
\mathbb{E} \left[ \prod_{i \in M} X^{a_i}_i \middle| \bm{X}_O \right],
\end{equation}
where $a_i$ are non-negative integers.

Our strategy will make use of two key facts.
First, Gaussian random vectors are closed under conditioning, hence $\bm{X}_M|\bm{X}_O$ is itself a Gaussian random vector whose parameters are functions of $\bm{\mu}$ and $\bm{\Sigma}$.
Second, arbitrary product moments of random vectors can generally be found by appropriately differentiating their moment-generating function \cite{Keener2010}.
Using the Gaussian moment-generating function in this way will remain a key tool for the rest of the theoretical development of this algorithm, so we will state the procedure in the following lemma.

\begin{lemma}[Gaussian Product Moments]
\label{lem:gaussian_mgf}
Let $\bm{X}$ be a Gaussian random vector distributed as $\bm{X} \sim \mathcal{N}_p(\bm{\mu}, \bm{\Sigma})$.
Then any product moment of the form $\mathbb{E} [ \prod_{i = 1}^p X^{a_i}_i ]$ can be expressed as a function of $\bm{\mu}$ and $\bm{\Sigma}$.
\end{lemma}

\begin{proof}
This follows from a straightforward use of the Gaussian moment-generating function, which is:
\begin{equation}
M_{\bm{X}}(\bm{t}) = \exp\left(\bm{t}^T \bm{\mu} + \dfrac{1}{2} \bm{t}^T \bm{\Sigma} \bm{t}\right).
\end{equation}
Then by the moment calculation property, any arbitrary product moment can be calculated with:
\begin{equation} \label{e_mgf}
\dfrac{\partial^a}{\prod_{i = 1}^p \partial t_i^{a_i}} M_{\bm{X}}(\bm{t}) \lvert_{\bm{t} = 0} = \mathbb{E} \left[ \prod_{i = 1}^{p} X^{a_i}_i \right],
\end{equation}
where $a = \sum_{i = 1}^p a_i$ and all $a_i$ take non-negative integer values.
\end{proof}

From here, the expectation in the form of Equation \ref{e_condProd} can be obtained by seeing that the parameters of $\bm{X}_M|\bm{X}_O$ are:
\begin{equation} \label{e_condGauss}
\begin{aligned}
\bm{\mu}_c &= \bm{\mu}_M + \bm{\Sigma}_{MO} \bm{\Sigma}_O^{-1}(\bm{x}_O - \bm{\mu}_O) \\
\bm{\Sigma}_c &= \bm{\Sigma}_M - \bm{\Sigma}_{MO} \bm{\Sigma}_O^{-1} \bm{\Sigma}_{OM},
\end{aligned}
\end{equation}
which follows from the well known parameterization of conditioning on Gaussian random vectors.
Then by applying Lemma \ref{lem:gaussian_mgf} on $\bm{X}_M|\bm{X}_O$, we obtain any of its product moments in terms of $\bm{\mu}$ and $\bm{\Sigma}$.
Thus, the latter 8 entries of $\mathbb{E}\left[ \bm{T}(\bm{U}) | \bm{U}_O \right]$ can be written in terms of $\bm{\theta}$ analytically.

Among the remaining 4 sufficient statistics, we turn our attention to $Y$, $Y X_j$ and $Y X_j X_k$.
Notice that we can consider a general expression that encapsulates the expectation of all three of these statistics by writing them as $\mathbb{E}[Y X^a_j X^b_k | \bm{X}_O ]$ for $a, b \in \{0, 1\}$.
Then we can re-write this quantity as:
\begin{equation} \label{eq:mdp1_yxjxk}
\mathbb{E}[Y X^a_j X^b_k | \bm{X}_O] = \mathbb{E} [\bm{d}(\bm{X})^T \bm{\beta} X^a_j X^b_k | \bm{X}_O],
\end{equation}
which follows from applying the law of total probability and Bayes' rule (see Appendix \ref{a_mdp1y1} for an explicit proof).
Noting that $\bm{d}(\bm{X})^T\bm{\beta}$ is a linear combination of products of $\bm{X}$, we apply Lemma \ref{lem:gaussian_mgf} with the linearity properties of the expectation operator to obtain $\mathbb{E}[Y X^a_j X^b_k | \bm{X}_O]$ in terms of $\bm{\theta}$.
Thus the solution for these expectations can be derived analytically as well.

Finally, the remaining expectation is $\mathbb{E}[Y^2 | \bm{X}_O]$.
This is derived as follows:
\begin{equation} \label{eq:mdp1_y2}
\begin{aligned}
\mathbb{E} \left[Y^2 | \bm{X}_O \right] &= \E \left[ \E\left[Y^2 \middle| \bm{X}_M, \bm{X}_O \right ] \middle| \bm{X}_O \right ] \\
&= \E \left[ \Var(Y | \bm{X}) + \E\left[Y | \bm{X} \right ]^2 \middle| \bm{X}_O \right ] \\
&= \E \left[ \sigma^2_{\epsilon} + (\bm{\beta}^T \bm{d}(\bm{X}))^2 \middle| \bm{X}_O \right ] \\
&= \sigma^2_{\epsilon} + \E \left[ \bm{\beta}^T \bm{d}(\bm{X}) \bm{d}(\bm{X})^T \bm{\beta} \middle| \bm{X}_O \right ] \\
&= \sigma^2_{\epsilon} + \E \left[ \sum_{i,j} \beta_i \beta_j (\bm{d}(\bm{X}) \bm{d}(\bm{X})^T)_{ij} \middle| \bm{X}_O \right ],
\end{aligned}
\end{equation}
where $(\bm{d}(\bm{X}) \bm{d}(\bm{X})^T)_{ij}$ refers to the $(i, j)$th element of $\bm{d}(\bm{X}) \bm{d}(\bm{X})^T$.
Once again, since each entry in the matrix $\bm{d}(\bm{X})\bm{d}(\bm{X})^T$ is a linear combination of products of $\bm{X}$, we can apply Lemma \ref{lem:gaussian_mgf} and the linearity of expectation to write $\mathbb{E} \left[Y^2 | \bm{X}_O \right]$ in terms of $\bm{\theta}$.
Thus finally, we have shown that all entries of $\E[ \bm{T}(\bm{U}) | \bm{U}_O]$ can be written as analytic functions of $\bm{\theta}$ under MDP 1.

\subsection{$Y$ is Observed, \boldmath{$X$} Admits No Fully Missing Products in $\boldmath{d}(\boldmath{X})$} \label{sec:mdp2}

MDP 2 considers the scenario where $Y$ is observed and $\bm{X}$ is patterned such that no product terms are fully missing.
Equivalently, we can say that $\bm{X}$ is patterned such that at least one $X_j$ is observed in every product term.
In this situation, $\bm{X}_M | Y, \bm{X}_O$ takes on a multivariate Gaussian distribution, and thus $\E[\bm{T}(\bm{U}) | \bm{U}_O]$ can be completely solved analytically.
To see why this is the case, let us re-write the analytical model in Equation \ref{regressionModel} under the assumptions of MDP 2.
First, note that we can separate terms by observed variables and missing variables:
\begin{equation}
\begin{aligned}
Y &= \bm{d}(\bm{X})^T \bm{\beta} + \epsilon \\
&= \beta_0 + \sum_{j = 1}^p \beta_j X_j + \sum_{j \neq k} \beta_{jk} X_j X_k + \epsilon \\
&= \beta_0 + \sum_{j \in O} \beta_j X_j + \sum_{j \in M} \beta_j X_j + \sum_{(j, k) \in O} \beta_{jk} X_j X_k + \sum_{j \in M, \, k \in O} \beta_{jk} X_j X_k + \epsilon.
\end{aligned}
\end{equation}
Then, we can regard all $\bm{X}_O$ as constants and absorb them into the intercept and product-term coefficients as follows:
\begin{equation}
\begin{aligned}
\tilde{\beta}_0 &\coloneqq \beta_0 + \sum_{j \in O} \beta_j X_j + \sum_{(j, k) \in O} \beta_{jk} X_j X_k \\
\tilde{\beta}_j &\coloneqq \beta_j + \beta_{jk}X_k, \text{ for } k \in O.
\end{aligned}
\end{equation}
This allows us to re-write the model only in terms of the missing variables as:
\begin{equation}
Y = \tilde{\beta}_0 + \sum_{j \in M} \tilde{\beta}_j X_j + \epsilon,
\end{equation}
from which we can write for any fixed $m \in M$:
\begin{equation}
X_m = \dfrac{Y - \tilde{\beta}_0 - \sum_{j \in M \setminus m} \tilde{\beta}_j X_j - \epsilon}{\tilde{\beta}_m}.
\end{equation}
Thus, any $X_m$ is a linear combination of other Gaussian random variables, therefore must be Gaussian itself.
Hence $\bm{X}_M | Y, \bm{X}_O$ follows a multivariate Gaussian distribution.
The derivation of the exact probability distribution $\Prob(\bm{X}_M | Y, \bm{X}_O)$ can be found in Appendix~\ref{sec:proof_mdp2}.

Since $Y$ is observed in this missing data pattern, $\E[\bm{T}(\bm{U}) | \bm{U}_O]$ only concerns product functions of $\bm{X}$.
Hence, we only need to apply Lemma \ref{lem:gaussian_mgf} to obtain these expectations, as $\bm{X}_M | Y, \bm{X}_O$ is a multivariate Gaussian.
Thus, under MDP 2, $\E[\bm{T}(\bm{U}) | \bm{U}_O]$ can be written as a function of $\bm{\theta}$ and solved analytically.

\subsection{$Y$ is Observed, \boldmath{$X$} Admits Fully Missing Products in $\boldmath{d}(\boldmath{X})$}

MDP 3 concerns the case where $Y$ is observed and $\bm{X}$ is patterned such that product terms are fully missing.
In this situation, the entries of $\E[\bm{T}(\bm{U}) | \bm{U}_O]$ may be difficult to derive analytically, or admit no closed-form.
This can be seen from the following characterization of $X_m$ for $m \in M$ under MDP 3:
\begin{equation}
\begin{aligned}
Y &= \bm{d}(\bm{X})^T \bm{\beta} + \epsilon \\
&= \beta_0 + \sum_{j = 1}^p \beta_j X_j + \sum_{j \neq k} \beta_{jk} X_j X_k + \epsilon \\
&= \beta_0 + \beta_m X_m + \sum_{j \neq m} \beta_j X_j + \sum_{j \neq m} \beta_{jm} X_j X_m + \sum_{j \neq k \neq m} \beta_{jk} X_j X_k + \epsilon \\
&= \beta_0 + \left(\beta_m + \sum_{j \neq m} \beta_{jm} X_j \right) X_m + \sum_{j \neq m} \beta_j X_j + \sum_{j \neq k \neq m} \beta_{jk} X_j X_k + \epsilon \\
\Rightarrow X_m &= \dfrac{Y - \beta_0 - \sum_{j \neq m} \beta_j X_j - \sum_{j \neq k \neq m} \beta_{jk} X_j X_k - \epsilon}{\beta_m + \sum_{j \neq m} \beta_{jm} X_j}.
\end{aligned}
\end{equation}
From here we can see that $X_m$ is a sum consisting of Gaussian ratio and product Gaussian ratio random variables.
The moments or moment generation function of such random variables are difficult to derive and not readily available.
Thus, this is the only missing data pattern for which numerical integration is used to obtain $\E[\bm{T}(\bm{U}) | \bm{U}_O]$.
That is, we approximate $\E[\bm{T}(\bm{U}) | \bm{U}_O]$ with
\begin{equation}
\begin{aligned}
\E[\bm{T}(\bm{U}) | \bm{U}_O] &= \dfrac{\int_{\bm{u}_M} \bm{T}(\bm{u}_M, \bm{u}_O) \Prob(\bm{u}_M | \bm{u}_O)\, d\bm{u}_M}{\int_{\bm{u}_M} \Prob(\bm{u}_M | \bm{u}_O)\, d\bm{u}_M} \\
&= \dfrac{\Prob(\bm{u}_O)}{\Prob(\bm{u}_O)} \dfrac{\int_{\bm{u}_M} \bm{T}(\bm{u}_M, \bm{u}_O) \Prob(\bm{u}_M, \bm{u}_O)\, d\bm{u}_M}{\int_{\bm{u}_M} \Prob(\bm{u}_{Mg}, \bm{u}_O)\, d\bm{u}_M} \\
&\approx \dfrac{\sum_{g = 1}^{G} \bm{T}(\bm{u}_{Mg}, \bm{u}_O) \Prob(\bm{u}_{Mg}, \bm{u}_O)}{\sum_{g = 1}^{G} \Prob(\bm{u}_{Mg}, \bm{u}_O)},
\end{aligned}
\end{equation}
where $\bm{u}_{Mg}$ is the $g$th grid point over the domain of $\bm{u}_M$ for numerical integration.
Note that the purpose of dividing by $1 = \int_{\bm{u}_M} \Prob(\bm{u}_M | \bm{u}_O)\, d\bm{u}_M$ is to cancel out $\Prob(\bm{u}_O)$ from the numerator.
This allows us to perform calculations in terms of $\Prob(\bm{u}_M, \bm{u}_O)$, rather than $\Prob(\bm{u}_M | \bm{u}_O)$, thus the latter need not be derived.

\subsection{Summary of Results}

We propose to construct a hybrid EM method that uses the analytic results derived in this Section for MDPs~1 and~2, and numerical integration for MDP~3.
A summary of these results are displayed in Table~\ref{tab:mdp_summary}.
This table essentially serves as a map that provides us with the method of calculation for any sufficient statistic's conditional expectation given any MDP.

\begin{table}[!hb]
\centering
\aboverulesep=0ex
\belowrulesep=0ex
\renewcommand{\arraystretch}{1.5}
\begin{tabular}{|c|c|c|c|c|c|}
\toprule
$\E[\bm{T}(\bm{U}) | \bm{U}_O]$ & $Y^2$ & $Y$ & $Y \overrightarrow{X^{}_j}^T$ & $Y \overrightarrow{X^{}_j X^{}_k}^T$ & $\mathbb{E} [ \prod_{i = 1}^p X^{a_i}_i ]$ \\
\midrule
MDP 1 & Eq.~\ref{eq:mdp1_y2} & \multicolumn{3}{c|}{Eq.~\ref{eq:mdp1_yxjxk}} & \multirow{2}{*}{Lemma~\ref{lem:gaussian_mgf}} \\
\cline{1-5}
MDP 2 & \multicolumn{4}{c|}{\multirow{2}{*}{Observed}} & \\
\cline{1-1} \cline{6-6}
MDP 3 & \multicolumn{4}{c|}{} & Numerical Integration \\
\bottomrule
\end{tabular}
\caption{Mosaic table of techniques to obtain $\E[\bm{T}(\bm{U}) | \bm{U}_O]$ by MDP.}
\label{tab:mdp_summary}
\end{table}

\section{Empirical Studies}

Given that the hybrid EM algorithm minimizes the use of numerical approximations, we now investigate the impact this has on data analysis.
This was done via two empirical studies: a traditional simulation study over a variety of conditions, and a smaller scale simulation study using real data with psychological measures.
We compared empirical parameter bias, empirical mean square error, confidence interval coverage rates, and computation time between the hybrid EM method and full numerical integration.

\subsection{Simulation Conditions} \label{sec:sim_conditons}

For the traditional simulation study, we sought to study estimator performance over a variety of realistic settings.
These settings were
\begin{itemize}
  \item Estimation method: hybrid EM (HYB) and full numerical integration (NI).
  \item Sample size ($n$): 50, 100, 250, 500, and 1000.
  \item Number of predictors ($p$): 3, 5, 10, and 15.
  \item Proportion of missingness ($\varphi_{\text{MIS}}$): 0.10, 0.20, and 0.30.
  \item Proportion among missing data from MDP 3 ($\varphi_{\text{MDP3}}$):  0, 0.05, 0.10, 0.15, 0.20, 0.25, 0.50, 0.75, and 1.
\end{itemize}
A fully factorial design was used for a total of $2 \times 5 \times 4 \times 3 \times 9 = 1080$ conditions.
The number of simulations per condition was set to 100.

\subsection{Data Generation} \label{sec:data_generate}

The parameters for $\bm{X}$ were generated in the following way:
\begin{equation}
\begin{aligned}
\bm{\mu} &\sim \mathcal{U}_p(-3, 3) \\
\bm{\Sigma} &= \bm{D}\bm{C}\bm{D},
\end{aligned}
\end{equation}
where $\bm{D}$ is a diagonal matrix of variances, with the diagonal distributed as $\mathcal{U}_p(1, \sqrt{3})$ and $\bm{C}$ is a constant correlation matrix with a unity diagonal and off-diagonal entries of 0.3.
Thus, each $\bm{\Sigma}$ is generated from the same underlying correlation matrix, but scaled accordingly with random variance entries.
Once these parameters were drawn, $\bm{X}$ was sampled from $\mathcal{N}_p(\bm{\mu}, \bm{\Sigma})$.

The parameters of the model were generated as follows.
First, the number of product-terms in the model was set to half the number of predictors rounded down, i.e., $\left \lfloor{\dfrac{p}{2}}\right \rfloor$.
Therefore the number of design variables in the model was
\begin{equation}
d = 1 + p + \left \lfloor{\dfrac{p}{2}}\right \rfloor,
\end{equation}
including the intercept.
$\bm{X}$ variables were chosen at random (uniformly) to form product-terms in $\bm{d}(\bm{X})$.
Then, an adjusted $R^2$ parameter and $\bm{\beta}$ vector were simulated using:
\begin{equation}
\begin{aligned}
R^2_a &\sim \mathcal{U}_1(0.1, 0.5) \\
\bm{\beta} &\sim \mathcal{U}_d(-3, 3).
\end{aligned}
\end{equation}
Given a sample of $\bm{x}$ vectors and a sampled $\bm{\beta}$, we can algebraically solve for a $\sigma^2_{\epsilon}$ such that the drawn $R^2_a$ is achieved (see Appendix~\ref{sec:gen_sig_e} for details).
This allows us to draw errors terms with $\epsilon \sim \mathcal{N}(0, \sigma^2_{\epsilon})$ and finally calculate the outcome with $Y = \bm{d}(\bm{X})^T \bm{\beta} + \epsilon$.

Once $\bm{X}$ and $Y$ were generated, the observed data indicator $\bm{R}$ was generated under a MAR mechanism.
This was done by randomly selecting a non-product variable in $\bm{X}$ (with uniform probability) to serve as an always-observed ``missingness anchor'' variable designated $X_a$, to determine missingness in all other variables.
This ensured the MAR assumption was always met.
Using an intermediate latent propensity variable based on $X_a$, we assigned MDPs~1, 2, and~3 such that both  $\varphi_{\text{MIS}}$ and $\varphi_{\text{MDP3}}$ were met according to the simulation condition and MDPs~1 and~2 were divided equally among the remaining missing cases.
The exact mathematical details of this procedure can be found in Appendix~\ref{sec:gen_mis}.


\subsection{Estimation Methods}

The EM algorithm outlined in Sections~\ref{sec:em_alg} and~\ref{sec:em_prod} was used to estimate the simulated datasets via the HYB and NI methods.
The true model specification that was generated along with the data was assumed to be known during estimation.
The HYB method used this information to categorize cases according to MDPs 1, 2, and 3 as defined in Section~\ref{sec:mdp}.
For cases under MDPs 1 and 2, the \textit{E}-step was calculated analytically according to the methods described in Sections~\ref{sec:mdp1} and~\ref{sec:mdp2}, respectively.
Numerical integration was used by the HYB method for cases under MDP 3, and by the NI method for all MDPs. The minimum number of integration points was set to be as close to 1,000, as possible.
Following the method used by \citeA{Robitzsch2021}, we used Riemann midpoint numerical integration as exemplified in Figure~\ref{fig:ni_ex}.
Confidence intervals were computed via non-parametric bootstrap \cite{Efron1994} using 1,000 bootstrap samples.
To minimize programming and computational variations between the two methods, both the HYB and NI methods were programmed using as much overlapping code as possible.
This was done in the \texttt{R} language \cite{R2021}, primarily using the \texttt{Rcpp} \cite{Eddelbuettel2011} and \texttt{RcppArmadillo} \cite{Eddelbuettel2014} packages, and the \texttt{Armadillo} package \cite{Sanderson2016} for the \texttt{C++} language.

\subsection{Performance Metrics}

Estimator performance metrics were empirical bias, empirical mean square error (MSE), 95\% confidence interval coverage rates, and computation time elapsed.
For bias, MSE, and 95\% confidence interval coverage rates, their component-wise counterparts were collected per coefficient per dataset.
That is, we collected
\begin{equation} \label{eq:outcomes}
\begin{aligned}
\text{Deviation} &\coloneqq \hat{\beta}_j - \beta_j \\
\text{Square Error} &\coloneqq (\hat{\beta}_j - \beta_j)^2 \\
\text{Coverage} &\coloneqq I(\widehat{LL}_j \leq \beta_j \leq \widehat{UL}_j),
\end{aligned}
\end{equation}
where $\widehat{LL}_j$ and $\widehat{UL}_j$ are the 2.5\% and 97.5\% empirical percentiles over the bootstrapped samples for $\hat{\beta}_j$ and  $I(\cdot)$ is the indicator function.
The average of these quantities provide the empirical bias, empirical MSE, and 95\% confidence interval coverage rates, respectively.
We analyzed the performance metrics by plotting the marginal means and conducting descriptive linear regressions of our outcomes as a function of simulation settings: method, $n$, $p$, $\varphi_{\text{MIS}}$, and $\varphi_{\text{MDP3}}$.


\subsection{Results}

Marginal plots of empirical bias, empirical MSE, coverage rates and time outcomes are displayed in Figures~\ref{fig:mar_diff}, \ref{fig:mar_diff2}, \ref{fig:mar_cover}, and~\ref{fig:mar_time}, respectively.
The results of the linear regressions are displayed in Tables~\ref{tab:lin_dev}, \ref{tab:lin_sqe}, \ref{tab:lin_cov}, and~\ref{tab:lin_time}, also respectively.
For the regressions, method was indicator coded ($\text{HYB} = 0, \text{NI} = 1$) and interacted with all predictors.
The order variable was included as a control (via indicator coding), and is defined as the number of $\bm{X}$ variables that serve as the the coefficient's regressor.
This is zero for intercepts, one for first-order terms ($\text{Order}_1$), and two for product-order terms ($\text{Order}_2$).

\subsubsection{Empirical Bias}

The marginal plot on empirical bias shows that both HYB and NI are generally unbiased across all conditions.
The regression coefficients shows that bias does not vary linearly with any of the predictors for HYB, and only very small effects may exist for NI, if at all.
Since the maximum likelihood estimates of $\bm{\beta}$ are unbiased, the missing data mechanism is MAR, and all models are guaranteed to be correctly specified, we would expect any bias to come from the method of estimation.
Our results show very little empirical bias incurred by both methods.

\subsubsection{Empirical Mean Square Error}

The marginal plots show empirical MSE decreasing with $n$ for both HYB and NI.
This is congruent with maximum likelihood theory, as $\hat{\bm{\beta}}$ is known to be asymptotically consistent.
However, the effect is more pronounced for HYB than NI, with the difference between the two being smaller at lower $n$.
This is also reflected by the linear regression cofficients.

For $p$, the MSE increases in a nearly linear fashion, with NI being higher than HYB.
This can be attributed to the fact that both HYB and NI implement numerical integration in their algorithms, which loses accuracy as the number of variables increase.
Further, the difference between the two methods slightly increases as $p$ increases, which is also reflected by the descriptive linear regressions.
This can be accounted for by the fact that HYB uses less numerical integration than NI.

Both methods see an increase in MSE as $\varphi_{\text{MIS}}$ increases in a nearly linear fashion.
The variance of the MLE is well-known to increase with the proportion of missing data, also known as the missing information principle \cite{Orchard1972}.
Since MSE is the squared bias plus the variance of the estimate, it follows that MSE would increase with $\varphi_{\text{MIS}}$.
Once again, both the marginal plot and descriptive regression shows that the NI method has a higher rate of increase than HYB.
This may be attributable to the fact that NI requires more numerical integration than HYB as $\varphi_{\text{MIS}}$ increases.

The marginal plot of $\varphi_{\text{MDP3}}$ shows that the MSE of NI is greater than (or equal to) that of HYB.
The MSE of NI was fairly unaffected by changes in $\varphi_{\text{MDP3}}$, while the MSE of HYB increased with $\varphi_{\text{MDP3}}$.
The means are identical at $\varphi_{\text{MDP3}} = 1$, reflecting the fact that both methods use full numerical integration in this condition.

\subsubsection{Coverage Rate}

The marginal plots show increases in coverage rates toward 0.95 as $n$ increases, with HYB being close to 0.95 than NI over all $n$.
This is attributable to the fact that bootstrap approximations become more accurate as the number samples in the original data increase (Glivenko-Cantelli theorem; \citeNP{Tucker1959}).
However, with $p$, $\varphi_{\text{MIS}}$, and $\varphi_{\text{MDP3}}$, we observe little to no trends, with coverage rates at about 0.93 to 0.94, overall exhibiting slight under-coverage.
Across all three variables, HYB is generally closer to the nominal rate than NI.
The one exception is when $\varphi_{\text{MDP3}} = 1$ as both methods are identical.

\subsubsection{Time}

Both the marginal plots and regression coefficients show that the time to complete estimation increases with all four variables $n$, $p$, $\varphi_{\text{MIS}}$, and $\varphi_{\text{MDP3}}$.
For the variables $n$, $p$, and $\varphi_{\text{MIS}}$, there is very little difference between HYB and NI.
For $\varphi_{\text{MDP3}}$ however, we see that HYB is substantially faster than NI at lower $\varphi_{\text{MDP3}}$, with the difference between the two decreasing to zero at $\varphi_{\text{MDP3}} = 1$.

\subsubsection{Summary}

A high-level summary of the marginal trends is as follows.
\begin{itemize}
  \item Both methods exhibited very little bias.
  \item The MSE of HYB was less than or equal to NI across all conditions. The difference increased with $n$, $p$, and $\varphi_{\text{MIS}}$, and decreased with $\varphi_{\text{MDP3}}$.
  \item Under most conditions, HYB had coverage rates closer to the nominal rate than NI. Generally, the difference was within 1\%, and both methods slightly under-covered in all conditions.
  \item The time difference between the methods show that HYB is generally faster than NI, though the difference is small. The exception is when $\varphi_{\text{MDP3}} = 0$, where the difference is substantial, but shrinks to zero as $\varphi_{\text{MDP3}}$ approaches 1.
\end{itemize}

\clearpage

\FloatBarrier
\begin{figure}[t]
\centering
\includegraphics[width=\textwidth, keepaspectratio]{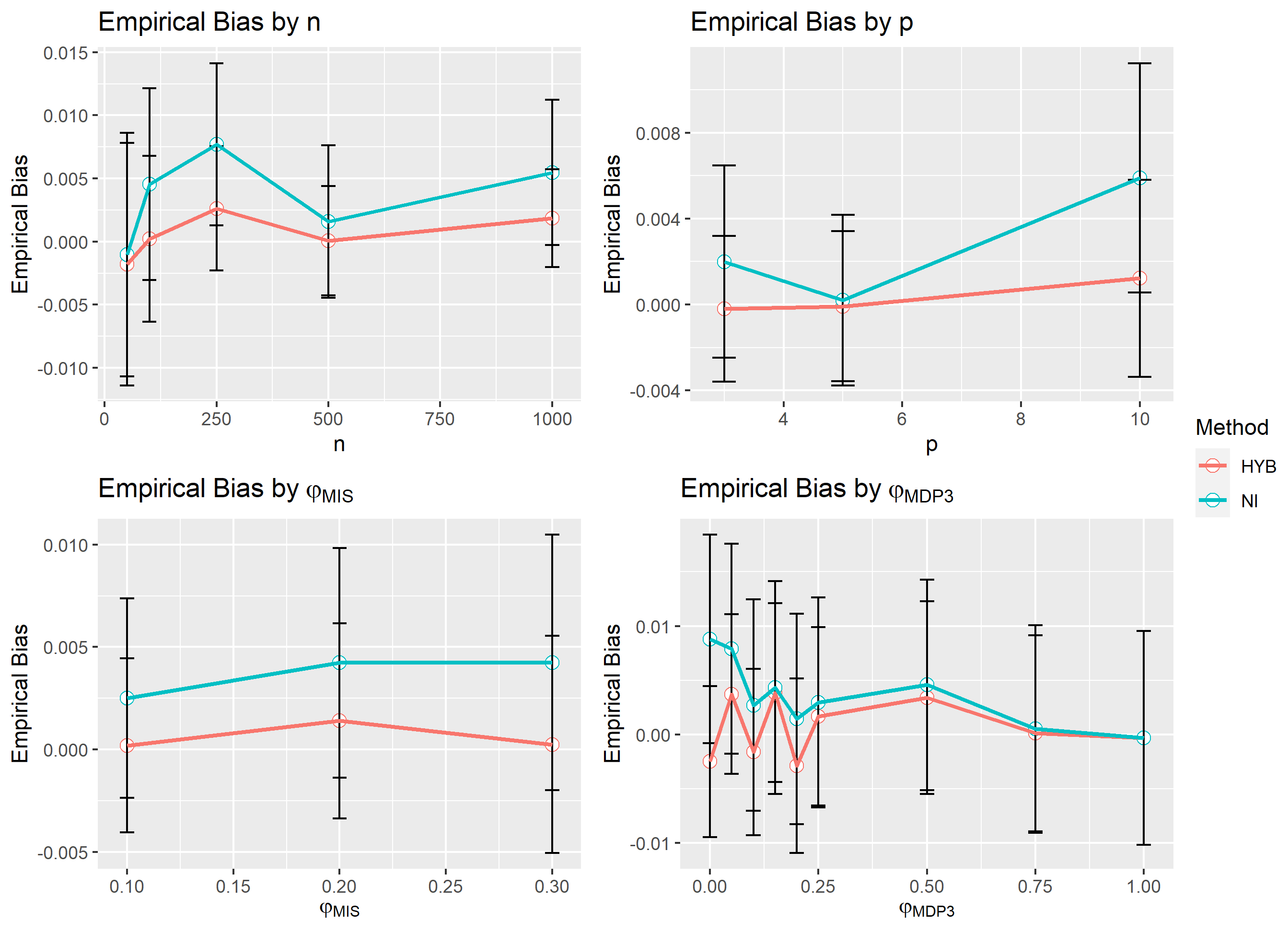} \\
\caption{Marginal view of bias by $n$, $p$, $\varphi_{\text{MIS}}$, and $\varphi_{\text{MDP3}}$, stratified by method.
Error bars indicate 95\% confidence intervals.}
\label{fig:mar_diff}
\end{figure}
\begin{table}[b]
\centering
\begin{tabular}{cS[table-format = 2.2e1]S[table-format = 2.2e1]S[table-format = 2.2e1]c}
\toprule
Coefficient & {Estimate} & {S. E.} & {$t$-value} & {$p$-value} \\
\midrule
Intercept  &  6.85e-03 & 7.50e-03 & 9.13e-01 & $\leq$0.361 \\
$n$  &  1.27e-06 & 4.34e-06 & 2.93e-01 & $\leq$0.770 \\
$p$  &  -1.30e-04 & 3.50e-04 & -3.72e-01 & $\leq$0.710 \\
$\rho_{\text{MIS}}$  &  -1.24e-02 & 1.85e-02 & -6.72e-01 & $\leq$0.502 \\
$\rho_{\text{MDP3}}$  &  1.29e-03 & 4.63e-03 & 2.79e-01 & $\leq$0.780 \\
NI  &  -1.03e-01 & 1.06e-02 & -9.70e+00 & $\leq$0.001 \\
NI $\times n$  &  3.29e-05 & 6.11e-06 & 5.37e+00 & $\leq$0.001 \\
NI $\times p$  &  9.85e-03 & 4.94e-04 & 2.00e+01 & $\leq$0.001 \\
NI $\times \rho_{\text{MIS}}$  &  1.90e-01 & 2.60e-02 & 7.30e+00 & $\leq$0.001 \\
NI $\times \rho_{\text{MDP3}}$  &  -6.58e-02 & 6.53e-03 & -1.01e+01 & $\leq$0.001 \\
\midrule
$\text{Order}_1$  &  -5.13e-03 & 5.82e-03 & -8.83e-01 & $\leq$0.377 \\
$\text{Order}_2$  &  -3.66e-03 & 6.19e-03 & -5.92e-01 & $\leq$0.554 \\
NI $\times \text{Order}_1$  &  1.24e-02 & 8.20e-03 & 1.51e+00 & $\leq$0.131 \\
NI $\times \text{Order}_2$  &  2.73e-02 & 8.73e-03 & 3.13e+00 & $\leq$0.002 \\
\bottomrule
\end{tabular}
\caption{Descriptive linear model of bias as a function of method, $n$, $p$, $\varphi_{\text{MIS}}$, $\varphi_{\text{MDP3}}$, and Order (control variable), along with their interactions with NI.}
\label{tab:lin_dev}
\end{table}

\FloatBarrier

\FloatBarrier
\begin{figure}[t]
\centering
\includegraphics[width=\textwidth, keepaspectratio]{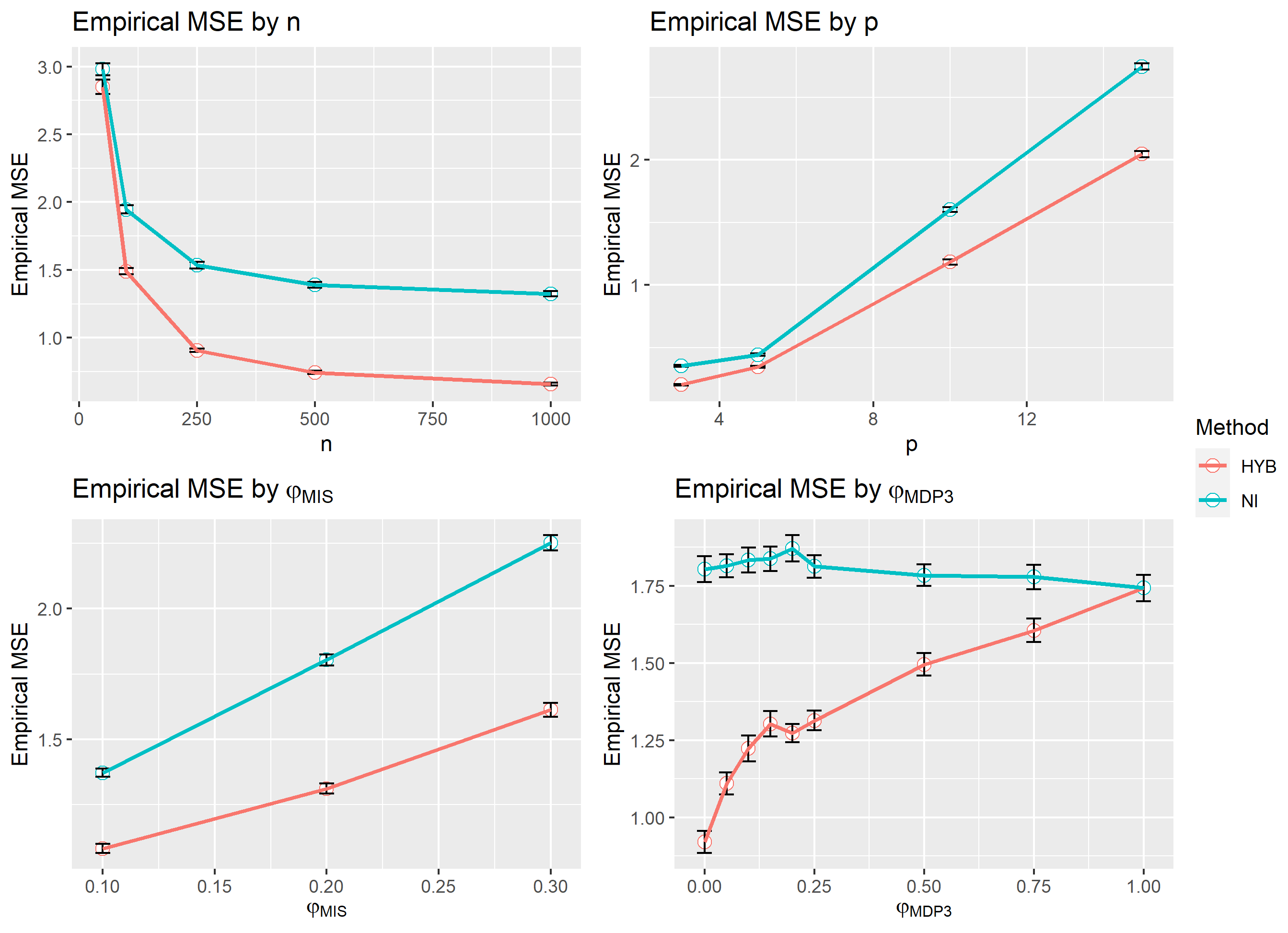} \\
\caption{Marginal view of mean square error by $n$, $p$, $\varphi_{\text{MIS}}$, and $\varphi_{\text{MDP3}}$, stratified by method.
Error bars indicate 95\% confidence intervals.}
\label{fig:mar_diff2}
\end{figure}
\begin{table}[b]
\centering
\begin{tabular}{cS[table-format = 2.2e1]S[table-format = 2.2e1]S[table-format = 2.2e1]c}
\toprule
Coefficient & {Estimate} & {S. E.} & {$t$-value} & {$p$-value} \\
\midrule
Intercept  &  2.13e+00 & 3.17e-02 & 6.71e+01 & $\leq$0.001 \\
$n$  &  -1.67e-03 & 1.83e-05 & -9.12e+01 & $\leq$0.001 \\
$p$  &  1.96e-01 & 1.48e-03 & 1.33e+02 & $\leq$0.001 \\
$\rho_{\text{MIS}}$  &  2.70e+00 & 7.80e-02 & 3.47e+01 & $\leq$0.001 \\
$\rho_{\text{MDP3}}$  &  6.54e-01 & 1.96e-02 & 3.35e+01 & $\leq$0.001 \\
NI  &  5.95e-01 & 4.47e-02 & 1.33e+01 & $\leq$0.001 \\
NI $\times n$  &  5.13e-04 & 2.58e-05 & 1.99e+01 & $\leq$0.001 \\
NI $\times p$  &  6.41e-02 & 2.09e-03 & 3.07e+01 & $\leq$0.001 \\
NI $\times \rho_{\text{MIS}}$  &  1.69e+00 & 1.10e-01 & 1.54e+01 & $\leq$0.001 \\
NI $\times \rho_{\text{MDP3}}$  &  -7.39e-01 & 2.76e-02 & -2.68e+01 & $\leq$0.001 \\
\midrule
$\text{Order}_1$  &  -3.30e+00 & 2.46e-02 & -1.34e+02 & $\leq$0.001 \\
$\text{Order}_2$  &  -3.21e+00 & 2.62e-02 & -1.23e+02 & $\leq$0.001 \\
NI $\times \text{Order}_1$  &  -1.23e+00 & 3.47e-02 & -3.56e+01 & $\leq$0.001 \\
NI $\times \text{Order}_2$  &  -1.14e+00 & 3.69e-02 & -3.08e+01 & $\leq$0.001 \\
\bottomrule
\end{tabular}
\caption{Descriptive linear model of mean square error as a function of method, $n$, $p$, $\varphi_{\text{MIS}}$, $\varphi_{\text{MDP3}}$, and Order (control variable), along with their interactions with NI.}
\label{tab:lin_sqe}
\end{table}

\FloatBarrier

\FloatBarrier
\begin{figure}[t]
\centering
\includegraphics[width=\textwidth, keepaspectratio]{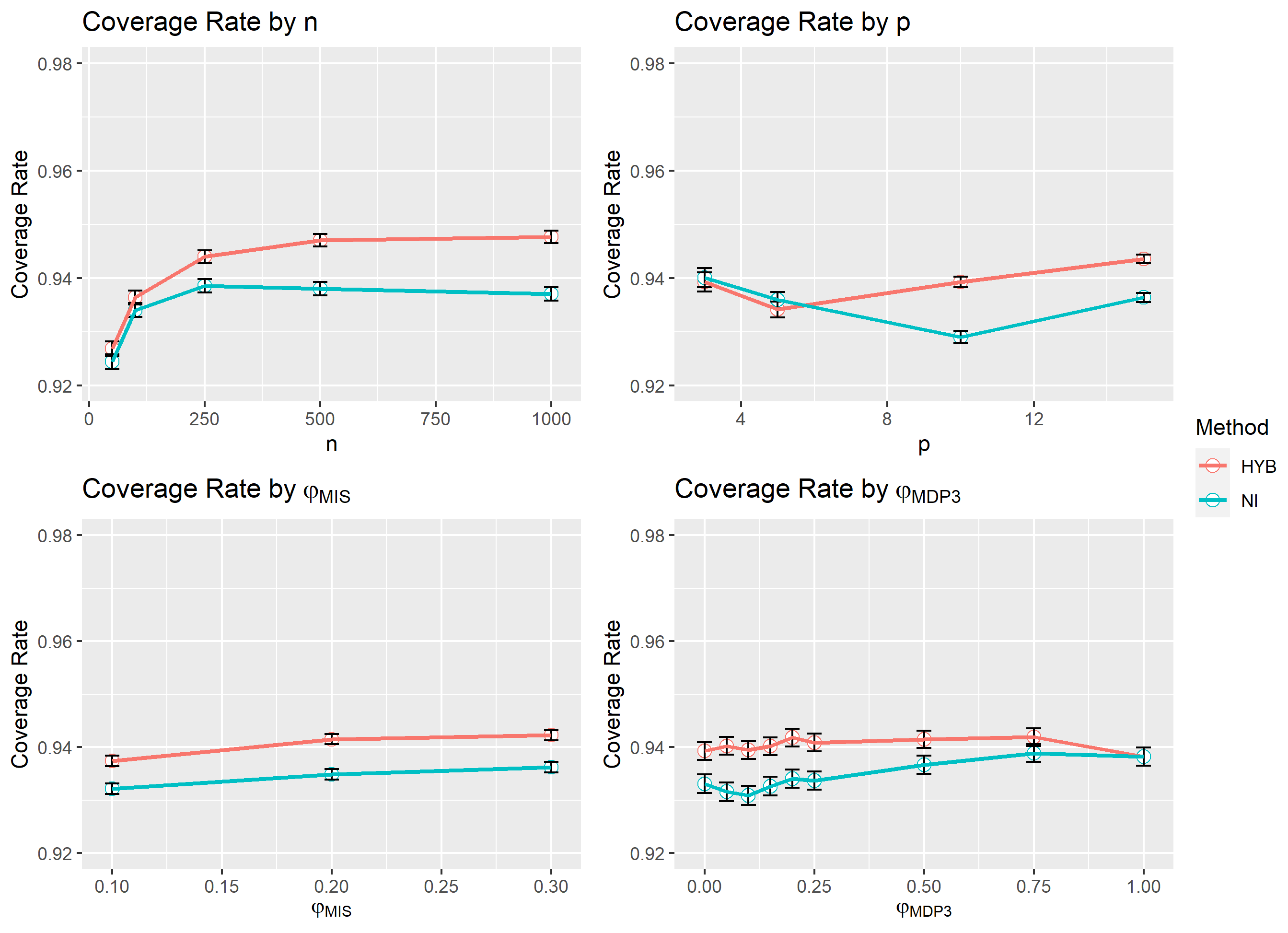} \\
\caption{Marginal view of coverage rate by $n$, $p$, $\varphi_{\text{MIS}}$, and $\varphi_{\text{MDP3}}$, stratified by method.
Error bars indicate 95\% confidence intervals.}
\label{fig:mar_cover}
\end{figure}
\begin{table}[b]
\centering
\begin{tabular}{cS[table-format = 2.2e1]S[table-format = 2.2e1]S[table-format = 2.2e1]c}
\toprule
Coefficient & {Estimate} & {S. E.} & {$t$-value} & {$p$-value} \\
\midrule
Intercept  &  9.23e-01 & 1.44e-03 & 6.39e+02 & $\leq$0.001 \\
$n$  &  1.73e-05 & 8.32e-07 & 2.08e+01 & $\leq$0.001 \\
$p$  &  6.78e-04 & 6.72e-05 & 1.01e+01 & $\leq$0.001 \\
$\rho_{\text{MIS}}$  &  2.44e-02 & 3.54e-03 & 6.88e+00 & $\leq$0.001 \\
$\rho_{\text{MDP3}}$  &  -2.22e-04 & 8.90e-04 & -2.50e-01 & $\leq$0.803 \\
NI  &  3.02e-03 & 2.04e-03 & 1.48e+00 & $\leq$0.139 \\
NI $\times n$  &  -9.16e-06 & 1.18e-06 & -7.78e+00 & $\leq$0.001 \\
NI $\times p$  &  -6.99e-04 & 9.51e-05 & -7.35e+00 & $\leq$0.001 \\
NI $\times \rho_{\text{MIS}}$  &  -3.85e-03 & 5.01e-03 & -7.68e-01 & $\leq$0.442 \\
NI $\times \rho_{\text{MDP3}}$  &  7.89e-03 & 1.26e-03 & 6.27e+00 & $\leq$0.001 \\
\midrule
$\text{Order}_1$  &  -1.03e-03 & 1.12e-03 & -9.20e-01 & $\leq$0.357 \\
$\text{Order}_2$  &  -1.98e-03 & 1.19e-03 & -1.66e+00 & $\leq$0.097 \\
NI $\times \text{Order}_1$  &  -8.37e-05 & 1.58e-03 & -5.30e-02 & $\leq$0.958 \\
NI $\times \text{Order}_2$  &  7.39e-04 & 1.68e-03 & 4.39e-01 & $\leq$0.661 \\
\bottomrule
\end{tabular}
\caption{Descriptive linear model of coverage rate as a function of method, $n$, $p$, $\varphi_{\text{MIS}}$, $\varphi_{\text{MDP3}}$, and Order (control variable), along with their interactions with NI.}
\label{tab:lin_cov}
\end{table}

\FloatBarrier

\FloatBarrier
\begin{figure}[t]
\centering
\includegraphics[width=\textwidth, keepaspectratio]{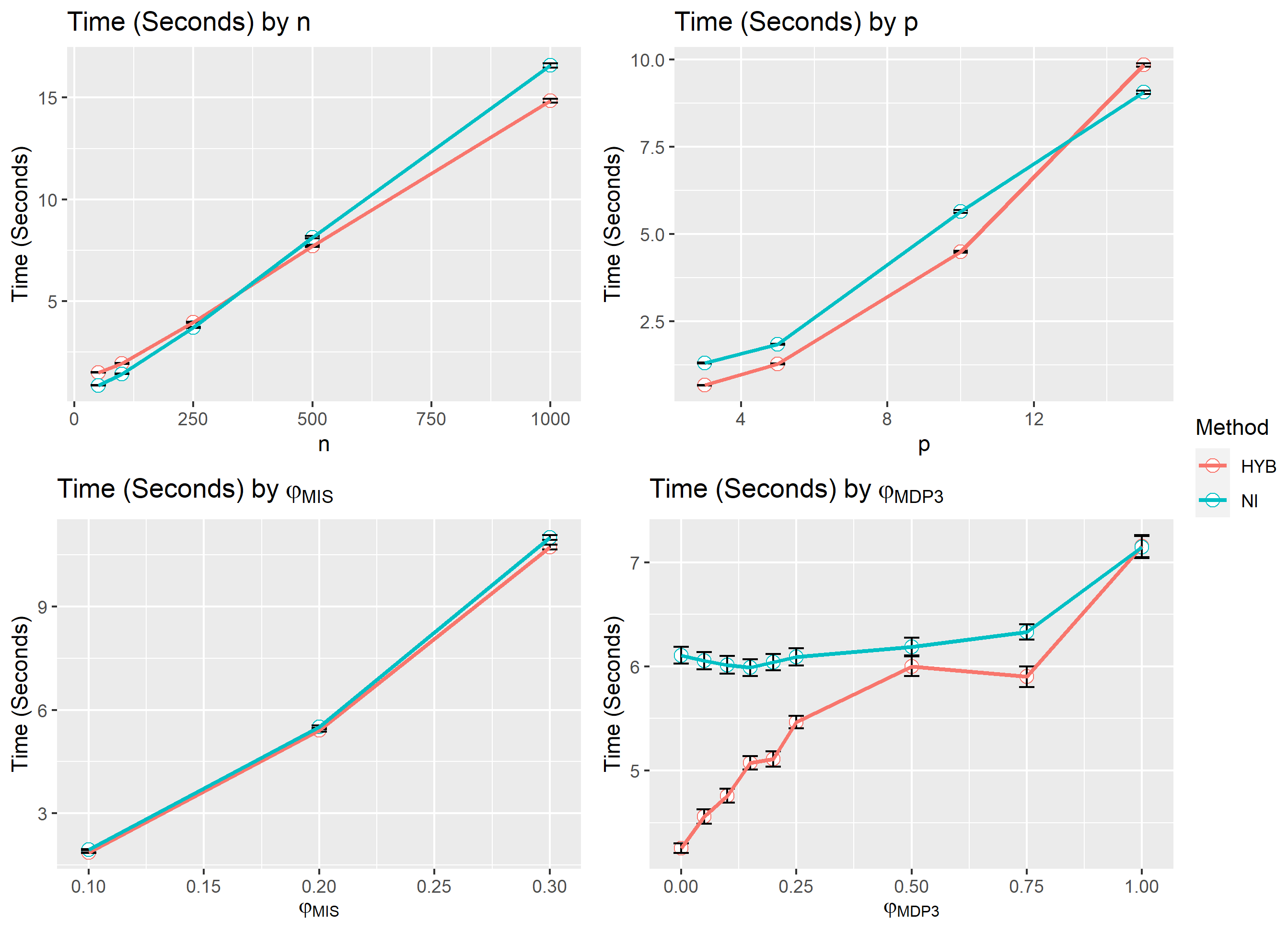} \\
\caption{Marginal view of average time by $n$, $p$, $\varphi_{\text{MIS}}$, and $\varphi_{\text{MDP3}}$, stratified by method.
Error bars indicate 95\% confidence intervals.}
\label{fig:mar_time}
\end{figure}
\begin{table}[b]
\centering
\begin{tabular}{cS[table-format = 2.2e1]S[table-format = 2.2e1]S[table-format = 2.2e1]c}
\toprule
Coefficient & {Estimate} & {S. E.} & {$t$-value} & {$p$-value} \\
\midrule
Intercept  &  -1.75e+01 & 5.33e-02 & -3.29e+02 & $\leq$0.001 \\
$n$  &  1.42e-02 & 3.08e-05 & 4.62e+02 & $\leq$0.001 \\
$p$  &  8.29e-01 & 2.48e-03 & 3.33e+02 & $\leq$0.001 \\
$\rho_{\text{MIS}}$  &  4.44e+01 & 1.31e-01 & 3.39e+02 & $\leq$0.001 \\
$\rho_{\text{MDP3}}$  &  1.48e+00 & 3.29e-02 & 4.50e+01 & $\leq$0.001 \\
NI  &  7.01e-01 & 7.54e-02 & 9.29e+00 & $\leq$0.001 \\
NI $\times n$  &  2.51e-03 & 4.35e-05 & 5.78e+01 & $\leq$0.001 \\
NI $\times p$  &  -1.44e-01 & 3.51e-03 & -4.09e+01 & $\leq$0.001 \\
NI $\times \rho_{\text{MIS}}$  &  9.15e-01 & 1.85e-01 & 4.95e+00 & $\leq$0.001 \\
NI $\times \rho_{\text{MDP3}}$  &  -7.87e-01 & 4.65e-02 & -1.69e+01 & $\leq$0.001 \\
\midrule
$\text{Order}_1$  &  -1.58e-01 & 4.13e-02 & -3.83e+00 & $\leq$0.001 \\
$\text{Order}_2$  &  -2.08e-01 & 4.40e-02 & -4.74e+00 & $\leq$0.001 \\
NI $\times \text{Order}_1$  &  1.23e-01 & 5.84e-02 & 2.11e+00 & $\leq$0.035 \\
NI $\times \text{Order}_2$  &  1.69e-01 & 6.22e-02 & 2.72e+00 & $\leq$0.007 \\
\bottomrule
\end{tabular}
\caption{Descriptive linear model of time as a function of method, $n$, $p$, $\varphi_{\text{MIS}}$, $\varphi_{\text{MDP3}}$, and Order (control variable), along with their interactions with NI.}
\label{tab:lin_time}
\end{table}

\FloatBarrier

\subsection{Real Data Study}

For our second simulation study, datasets were generated by bootstrapping samples using a real dataset as a basis, where simulated missingness was artificially inserted.
We analyzed measures of psychopathology from the Adolescent Brain Cognitive Development (ABCD) Study (https://abcdstudy.org).
The ABCD is a large, multi-site study, whose data are publicly available \cite{Volkow2018}, which was approved by the institutional review boards of the participating sites \cite{Clark2018}.
To avoid potential clustering effects by site and to reduce the sample size to a more realistic scale, one site was selected randomly with uniform probabity to provide the basis of our data ($n = 604$).

We considered a linear model of conduct disorder (CD) as a function of attention deficit hyperactivity disorder (ADHD), depression (DEP), their product-term (ADHD $\times$ DEP), controlled for by anxiety (ANX) and oppositional defiant disorder (ODD).
Previous work has shown comorbidity among these variables \cite{Angold1999, Jensen1997}.
Measures were taken using summary scores of the child behavior checklist (CBCL; \citeNP{Achenbach2001}).
The selected site incidentally had no missing data on these variables.

For each bootstrapped data set, missingness was inserted using the same MAR generating procedure as the traditional simulation study.
We fixed $\varphi_{\text{MIS}} = 0.3$ and $\varphi_{\text{MDP3}} = 1/3$, which set equal proportions among all MDPs.
For performance metrics, we collected deviation and square error as in Equation~\ref{eq:outcomes} from the HYB and NI methods.
The least squares estimates of the original non-missing data set was considered as the true parameters.

Boxplots of the results are displayed in Figure~\ref{fig:rd_boxplot}, and averages are shown in Table~\ref{tab:rd_results}.
On average, the HYB method had biases closer to zero than NI across all parameter estimates.
Likewise, the HYB method also had MSEs closer to zero than NI across all parameter estimates except for the product-term coefficient, where NI had a marginally smaller MSE by $3.9 \times 10^{-4}$.

\clearpage

\FloatBarrier
\begin{figure}[t]
\centering
\includegraphics[width=\textwidth, keepaspectratio]{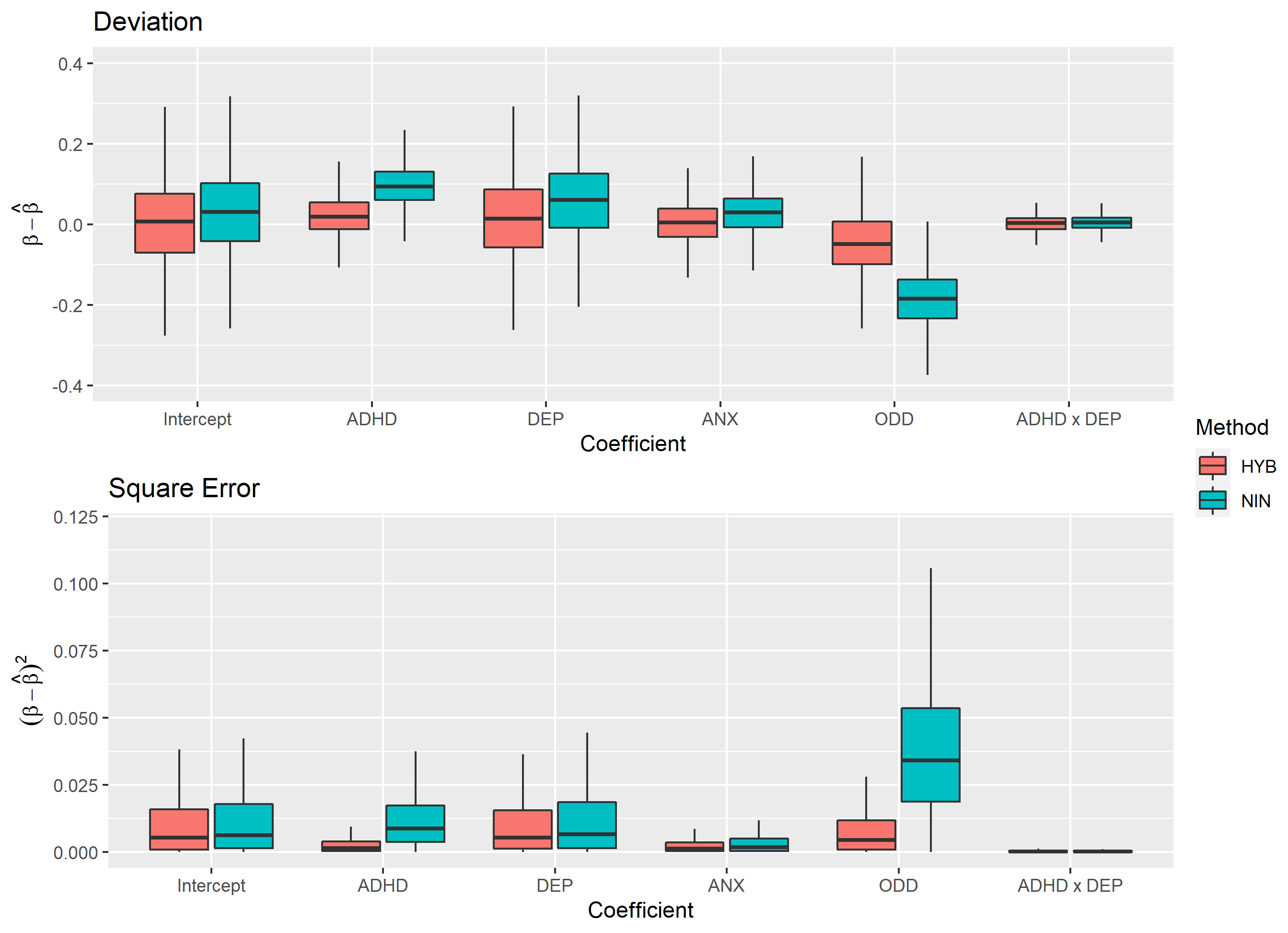} \\
\caption{Boxplots of deviations and square errors of the parameter estimates in the real data simulation.}
\label{fig:rd_boxplot}
\end{figure}
\begin{table}[t]
\centering
\begin{tabular}{cS[table-format = 2.2e1]S[table-format = 2.2e1]S[table-format = 2.2e1]S[table-format = 2.2e1]}
\toprule
 & \multicolumn{2}{c}{{Bias}} & \multicolumn{2}{c}{{MSE}} \\
\midrule
Coefficient       & {HYB}     & {NI}     & {HYB}    & {NI}     \\
\midrule
Intercept         &  0.0074   &  0.0299  & 0.0123   & 0.0131   \\
ADHD              &  0.0217   &  0.0974  & 0.00322  & 0.0122   \\
DEP               &  0.0141   &  0.0603  & 0.0117   & 0.0138   \\
ANX               &  0.0048   &  0.0293  & 0.00286  & 0.00382  \\
OPP               & -0.0463   & -0.184   & 0.00861  & 0.0393   \\
ADHD $\times$ DEP &  0.00161  &  0.00317 & 0.000383 & 0.000344 \\
\midrule
Means             &  0.000543 &  0.00603 & 0.00651  & 0.0138   \\
\bottomrule
\end{tabular}
\caption{Bias and MSE averaged within method and coefficient.}
\label{tab:rd_results}
\end{table}
\FloatBarrier

\section{Discussion}

In this research, we sought to improve the EM algorithm for linear models with two-way product terms by deriving analytic $E$-steps for large classes of MDPs.
These derivations were used to develop a hybrid approach to the EM algorithm, where analytic $E$-steps were used whenever possible and numerical integration was used otherwise.
Through two simulation studies, we showed that the hybrid approach outperformed, or performed as well as, the existing method of full numerical integration on several performance metrics.
The difference between the two methods diminished as the proportion of missing data that require numerical integration for the hybrid approach increased ($\varphi_{\text{MDP3}}$).
Otherwise, the degree to which the hybrid method out performed full numerical integration grew with $n$, $p$, and $\varphi_{\text{MIS}}$.

\subsection{Limitations and Future Extensions}

The current research focused on regression models with two-way product-terms.
While this is a very common model in the psychological sciences, occasionally higher-power polynomials
and multi-way product-terms are called for.
The approach of deriving analytic $E$-steps for two-way product term models generalizes in a straightforward manner to all polynomial forms.
Higher powers of predictors are readily accommodated by Lemma~\ref{lem:gaussian_mgf}, since the Gaussian moment-generating function can be used to obtain any non-negative integer moment.
Arbitrarily large product-terms can also be accommodated, although more limited due to the nature of MDPs~2 and~3.
MDP~2 relies on the fact that any missing $X_j$ is linear with respect to all other missing $X_{-j}$ conditional on the observed data.
This holds as long as both factors of any product-term are not missing simultaneously.
For product-terms with more than two factors, the conditional linearity strategy could only accommodate one factor of the product-term to be missing.
Otherwise, the missing predictor would be conditionally non-linear and numerical integration would be required.
Therefore, as the number of factors in a product-term increases, it becomes less likely that the conditional linearity strategy can be used, diminishing the advantages of the hybrid approach.

Other extensions can be made to discrete predictors.
Ordinal data can be accommodated using existing methods for estimating polychoric correlations assuming the data are a result of discretizing Gaussian variables \cite{Olsson1979}.
For nominal data, techniques from Gaussian mixture models may also be incorporated into estimation.
Such extensions would greatly advance the applicability of the hybrid method as both ordinal and nominal data are commonly used.

For numerical integration, we focused exclusively on the Riemann midpoint method.
Both the speed and accuracy of numerical integration may be improved via other approximating functions (e.g., \citeNP{Liu1994}), and/or by adaptive methods (e.g., \citeNP{Rabe-Hesketh2002}).
Monte Carlo methods of integration also exist, including Gibbs and Metropolis-Hasting variants \cite{Levine2001, Wei1990}, which can handle high-dimensions more efficiently. However, these methods require more samples as the number of EM iterations progress, so that Monte Carlo error does not interfere with EM convergence assessment \cite{Booth1999}.
The investigation of these issues could lead to improvements in both the hybrid and numerical integration approaches.

A final area of further research pertain to implementation of standard errors.
While we employed bootstrap standard errors in the current research, other methods exist that utilize numerical differentiation of either the EM operation \cite{Meng1991} or the observed Fisher score function \cite{Jamshidian1999}.
While these methods are able to directly approximate the observed-data Hessian matrix, their computational complexity is on the order of the number of parameters squared \cite{Jamshidian2000}. This makes computation time a practical issue for larger models.
Direct analytic derivation of the observed-data Hessian is also possible, but often deemed too tedious to undertake \cite{Hartley1958}.
However, this may be feasible with a proper MDP taxonomy, such as MDPs~1, 2, and~3 described in Section~\ref{sec:mdp}, or with symbolic computation methods.

\subsection{Implications and Conclusions}

An interesting application of the hybrid approach pertains to planned missingness designs.
In planned missingness designs, there is a large degree of control over missing data patterns.
The benefits of the hybrid approach can be fully realized if the missing data patterns are manipulated to completely avoid MDP~3.
This can be done by designing a planned missingness study to contain a block of items to be delivered to all participants (a non-missing block, e.g., item set X in a 3-form design; \citeNP{Graham2006}).
Then MDP~3 can be completly avoided by assigning one of the factors of each product-term to the non-missing data block.

Ultimately, the accuracy of scientifically inferential quantities, such as parameter estimates and coverage intervals are fundamental aspects for the applied use of any methodology.
Improvements in these areas directly impact the ability to detect effects accurately.
The wide use of product-term regression models and the prevalence of missing data issues makes our hybrid EM method highly applicable to problems in the psychological sciences.
A final step of this research would be to develop and disseminate a user-friendly software implementation, allowing accessible use of the hybrid EM method to the end-user researcher.

\newpage
\bibliographystyle{apacite}
\bibliography{references}

\newpage
\appendix
\renewcommand{\thesection}{A.\arabic{section}}
\renewcommand\thefigure{\thesection.\arabic{figure}}
\renewcommand\thetable{\thesection.\arabic{table}}

\begin{center}
\LARGE
\hfill\\
\hfill\\
\textbf{Appendix}
\hfill\\
\hfill\\
\end{center}

\section{Sufficient Statistics of the Joint Model} \label{sec_ss}

In this section we derive the vector of sufficient statistics $\bm{T}(\bm{U})$.
Noting that $\Prob(\bm{U}) = \Prob(Y | \bm{X}) \Prob(\bm{X})$, and $\Prob(Y | \bm{X})$ and $\Prob(\bm{X})$ are both exponential family distributions, we can derive the sufficient statistics in a piecewise manner.
It is well known that $\bm{T}(\bm{X}) = \begin{bmatrix} \bm{X}^T & \overrightarrow{X_j X_k}^T & \overrightarrow{X^2_j}^T \end{bmatrix}^T$ since $\Prob(\bm{X})$ is Gaussian.
Thus, what remains is determining the entries of $\bm{T}(\bm{U})$ that are contributed by $\Prob(Y | \bm{X})$.
Since $\Prob(Y | \bm{X})$ is also Gaussian, we can determine its sufficient statistics by examining the exponential as follows:
\begin{equation}
\begin{aligned}
\Prob(Y | \bm{X}) &= \dfrac{1}{\sqrt{2 \pi \sigma^2}} \exp \left( -\dfrac{1}{2 \sigma^2} \left( Y - \bm{d}(\bm{X})^T \bm{\beta}\right)^2 \right) \\
&= \dfrac{1}{\sqrt{2 \pi \sigma^2}} \exp \left( -\dfrac{1}{2 \sigma^2} \left( Y^2 - 2Y\bm{d}(\bm{X})^T \bm{\beta} + \bm{\beta}^T \bm{d}(\bm{X})\bm{d}(\bm{X})^T \bm{\beta} \right) \right).
\end{aligned}
\end{equation}
For which we have $Y^2$, $Y \bm{d} (\bm{X})^T$, and $\bm{d}(\bm{X}) \bm{d}(\bm{X})^T$.
$Y^2$ is self-explanatory and $Y \bm{d} (\bm{X})^T = \begin{bmatrix} Y & Y \bm{X}^T & Y \overrightarrow{X_j X_k}^T \end{bmatrix}^T$.
Then what remains for us to analyze is the $\bm{d}(\bm{X}) \bm{d}(\bm{X})^T$ term.
By simply checking the elements of this outer product we obtain:
\begin{equation}
\bm{d}(\bm{X}) \bm{d}(\bm{X})^T =
\begin{bmatrix}
1 & X^T & \overrightarrow{X_j X_k}^T \\
X & XX^T & X \overrightarrow{X_j X_k}^T \\
\overrightarrow{X_j X_k} & \overrightarrow{X_j X_k} X^T & \overrightarrow{X_j X_k} \overrightarrow{X_j X_k}^T
\end{bmatrix}.
\end{equation}
Then we can see that:
\begin{equation}
\begin{aligned}
\text{vec}(\overrightarrow{X_j X_k} X^T) &= \bm{P}_1 \begin{bmatrix} \overrightarrow{X_j X_k X_l}^T & \overrightarrow{X^2_j X_k}^T \end{bmatrix}^T \\
\text{vech}(\overrightarrow{X_j X_k} \overrightarrow{X_j X_k}^T) &= \bm{P}_2 \begin{bmatrix} \overrightarrow{X_j X_k X_l X_m}^T & \overrightarrow{X^2_j X_k X_l}^T & \overrightarrow{X^2_j X^2_k}^T \end{bmatrix}^T, \\
\end{aligned}
\end{equation}
where $\text{vec}(\cdot)$ and $\text{vech}(\cdot)$ are the vectorization and half-vectorization functions, respectively.
$\bm{P}_1$ and $\bm{P}_2$ are permutation matrices that yields the appropriate ordering.
Finally, we combine the sufficient statistics of $\Prob(Y | \bm{X})$ and $\Prob(\bm{X})$, we obtain:
\begin{equation}
\bm{T}(\bm{U}) = \left[ \begin{smallmatrix}
Y & Y \overrightarrow{X^{}_j}^T & Y^2 & Y \overrightarrow{X^{}_j X^{}_k}^T & \overrightarrow{X^{}_j}^T & \overrightarrow{X^{}_j X^{}_k}^T & \overrightarrow{X^2_j}^T & \overrightarrow{X^{}_j X^{}_k X^{}_l}^T & \overrightarrow{X^2_j X^{}_k}^T & \overrightarrow{X^{}_j X^{}_k X^{}_l X^{}_m}^T & \overrightarrow{X^2_j X^{}_k X^{}_l}^T & \overrightarrow{X^2_j X^2_k}^T
\end{smallmatrix} \right]^T.
\end{equation}

\section{Maximization of the $Q$-function} \label{sec:q_max}

The $Q$-function for a sample can be written as follows:
\begin{equation}
\label{eq:qfun_sample}
\begin{aligned}
Q_{\bm{\theta}^{(t)}}(\bm{\theta}) &= \sum_{i = 1}^n \Es{\log \Prob_{\bm{\theta}}(\bm{u}_i)} \\
&= \sum_{i = 1}^n \Es{\log \Prob_{\bm{\theta}}(y_i | \bm{x}_i) + \log \Prob_{\bm{\theta}}(\bm{x}_i)}.
\end{aligned}
\end{equation}
Maximizing with respect to $\bm{\beta}$, the relevant terms of the $Q$-function are:
\begin{equation}
\begin{aligned}
Q_{\bm{\theta}^{(t)}}(\bm{\beta}) &= -\dfrac{1}{2\sigma^2_{\epsilon}} \sum_{i = 1}^n \Es{(y_i - \bm{d}(\bm{x}_i)^T \bm{\beta})^2} + c \\
&= -\dfrac{1}{2\sigma^2_{\epsilon}} \sum_{i = 1}^n \Es{\bm{\beta}^T \bm{d}(\bm{x}_i) \bm{d}(\bm{x}_i)^T \bm{\beta} - 2y_i\bm{d}(\bm{x}_i)^T \bm{\beta}} + c'.
\end{aligned}
\end{equation}
The first partial derivative is then
\begin{equation}
\pder{Q_{\bm{\theta}^{(t)}}(\bm{\beta})}{\bm{\beta}} = -\dfrac{1}{2\sigma^2_{\epsilon}} \sum_{i = 1}^n \Es{ 2 \bm{d}(\bm{x}_i) \bm{d} (\bm{x}_i)^T \bm{\beta} - 2 y_i \bm{d}(\bm{x}_i)},
\end{equation}
which we set to zero and obtain the first-order condition of
\begin{equation}
\sum_{i = 1}^n \Es{ 2 \bm{d}(\bm{x}_i) \bm{d} (\bm{x}_i)^T}\bm{\beta} = \sum_{i = 1}^n \Es{2 y_i \bm{d}(\bm{x}_i)}.
\end{equation}
Thus our estimate of $\bm{\beta}$ is
\begin{equation}
\label{eq:hat_beta}
\hat{\bm{\beta}} = \left(\sum_{i = 1}^n \Es{\bm{d}(\bm{x}_i) \bm{d} (\bm{x}_i)^T} \right)^{-1} \sum_{i = 1}^n \Es{y_i \bm{d}(\bm{x}_i)}.
\end{equation}

Maximizing with respect to $\sigma^2_{\epsilon}$, we return to Equation~\ref{eq:qfun_sample} and once again collect the relevant terms as
\begin{equation}
Q_{\bm{\theta}^{(t)}}(\sigma^2_{\epsilon}) = -\dfrac{n}{2}\log(\sigma^2_{\epsilon}) - \dfrac{1}{2\sigma^2_{\epsilon}} \sum_{i = 1}^n \Es{(y_i - \bm{d}(\bm{x}_i)^T\bm{\beta})^2} + c.
\end{equation}
The first partial derivative is then
\begin{equation}
\pder{Q_{\bm{\theta}^{(t)}}(\sigma^2_{\epsilon})}{\sigma^2_{\epsilon}} = -\dfrac{n}{2\sigma^2_{\epsilon}} +  \dfrac{1}{2\sigma^4_{\epsilon}} \sum_{i = 1}^n \Es{(y_i - \bm{d}(\bm{x}_i)^T \bm{\beta})^2},
\end{equation}
which yields the first-order condition and estimate of $\sigma^2_{\epsilon}$ as
\begin{equation}
\begin{aligned}
\dfrac{n}{2\sigma^2_{\epsilon}} &= \dfrac{1}{2\sigma^4_{\epsilon}} \sum_{i = 1}^n \Es{(y_i - \bm{d}(\bm{x}_i)^T\bm{\beta})^2} \\
\Rightarrow \hat{\sigma}^2_{\epsilon} &= \dfrac{\sum_{i = 1}^n \Es{(y_i - \bm{d}(\bm{x}_i)^T\hat{\bm{\beta}})^2}}{n}.
\end{aligned}
\end{equation}
To express this quantity strictly in terms of the expected sufficient statistics, we expand the quadratic term to obtain
\begin{equation}
\label{eq:hat_sigma_e_2}
\begin{aligned}
\hat{\sigma}^2_{\epsilon} &= \dfrac{\sum_{i = 1}^n \Es{(y_i - \bm{d}(\bm{x}_i)\hat{\bm{\beta}})^2}}{n} \\
&= \dfrac{\sum_{i = 1}^n \Es{y^2_i}}{n} - 2 \dfrac{\sum_{i = 1}^n \Es{y_i \bm{d}(\bm{x}_i)^T}}{n}\hat{\bm{\beta}} + \hat{\bm{\beta}}^T \dfrac{\sum_{i = 1}^n \Es{\bm{d}(\bm{x}_i) \bm{d}(\bm{x}_i)^T}}{n}\hat{\bm{\beta}}.
\end{aligned}
\end{equation}
For simplicity, let us express $\hat{\bm{\beta}}$ with the following shorthand notation:
\begin{equation}
\begin{aligned}
\bm{A} &\coloneqq \sum_{i = 1}^n \Es{\bm{d}(\bm{x}_i) \bm{d} (\bm{x}_i)^T} \\
\bm{b} &\coloneqq \sum_{i = 1}^n \Es{y_i \bm{d}(\bm{x}_i)} \\
\Rightarrow \hat{\bm{\beta}} &= \bm{A}^{-1} \bm{b}.
\end{aligned}
\end{equation}
Re-expressing Equation~\ref{eq:hat_sigma_e_2} with these quantities, we have:
\begin{equation}
\begin{aligned}
\hat{\sigma}^2_{\epsilon} &= \dfrac{\sum_{i = 1}^n \Es{y^2_i}}{n} - 2 \dfrac{\bm{b}^T}{n}\hat{\bm{\beta}} + \hat{\bm{\beta}}^T \dfrac{\bm{A}}{n}\hat{\bm{\beta}} \\
&= \dfrac{\sum_{i = 1}^n \Es{y^2_i}}{n} - 2 \dfrac{\bm{b}^T\bm{A}^{-1} \bm{b}}{n} +  \dfrac{\bm{b}^T \bm{A}^{-1} \bm{A} \bm{A}^{-1} \bm{b}}{n} \\
&= \dfrac{\sum_{i = 1}^n \Es{y^2_i}}{n} - \dfrac{\bm{b}\bm{A}^{-1} \bm{b}}{n} \\
\end{aligned}
\end{equation}
which finally yields
\begin{equation}
\begin{split}
\hat{\sigma}^2_{\epsilon} = \dfrac{\sum_{i = 1}^n \Es{y^2_i}}{n} - \dfrac{\sum_{i = 1}^n \Es{y_i \bm{d}(\bm{x}_i)}\left(\sum_{i = 1}^n \Es{\bm{d}(\bm{x}_i) \bm{d} (\bm{x}_i)^T} \right)^{-1}}{n} \\ \times \sum_{i = 1}^n \Es{y_i \bm{d}(\bm{x}_i)}
\end{split}
\end{equation}

Returning to Equation~\ref{eq:qfun_sample} and collecting the relevant terms for $\bm{\mu}$ we have
\begin{equation}
\begin{aligned}
Q_{\bm{\theta}^{(t)}}(\bm{\mu}) &= -\dfrac{1}{2} \sum_{i = 1}^n \Es{(\bm{x}_i - \bm{\mu})^T \bm{\Sigma}^{-1} (\bm{x}_i - \bm{\mu})} + c \\
&= -\dfrac{n}{2} \bm{\mu}^T \bm{\Sigma}^{-1} \bm{\mu} + \sum_{i = 1}^n \Es{ \bm{x}_i }^T \bm{\Sigma}^{-1} \bm{\mu} + c'.
\end{aligned}
\end{equation}
This yields a first partial derivative of
\begin{equation}
\pder{Q_{\bm{\theta}^{(t)}}(\bm{\mu})}{\bm{\mu}} = -n\bm{\Sigma}^{-1}\bm{\mu} + \bm{\Sigma}^{-1}  \sum_{i = 1}^n \Es{\bm{x}_i},
\end{equation}
and thus our first-order condition and estimator for $\bm{\mu}$ is
\begin{equation}
\begin{aligned}
n\bm{\Sigma}^{-1}\bm{\mu} &= \bm{\Sigma}^{-1}  \sum_{i = 1}^n \Es{\bm{x}_i} \\
\Rightarrow \hat{\bm{\mu}} &= \dfrac{ \sum_{i = 1}^n \Es{\bm{x}_i}}{n}
\end{aligned}
\end{equation}

Finally, for $\bm{\Sigma}$, we once again collect the relevant terms from the $Q$-function as follows:
\begin{equation}
Q_{\bm{\theta}^{(t)}}(\bm{\Sigma}) = -\dfrac{n}{2}\log\left|\bm{\Sigma}\right| -\dfrac{1}{2} \sum_{i = 1}^n \Es{(\bm{x}_i - \bm{\mu})^T \bm{\Sigma}^{-1} (\bm{x}_i - \bm{\mu})} + c,
\end{equation}
which yields a first partial derivative of
\begin{equation}
\pder{Q_{\bm{\theta}^{(t)}}(\bm{\Sigma})}{\bm{\Sigma}} = -\dfrac{n}{2}\bm{\Sigma}^{-1} + \dfrac{1}{2} \bm{\Sigma}^{-1} \left( \sum_{i = 1}^n \Es{(\bm{x}_i - \bm{\mu}) (\bm{x}_i - \bm{\mu})^T} \right) \bm{\Sigma}^{-1},
\end{equation}
where we used the identities of $\pder{\log\left|\bm{A}\right|}{\bm{A}} = \bm{A}^{-1}$ and $\pder{\bm{b}^T\bm{A}^{-1}\bm{c}}{\bm{A}} = -\left(\bm{A}^{-1} \bm{c} \bm{b}^T \bm{A}^{-1}\right)^T$ for an invertible matrix $\bm{A}$, and vectors $\bm{b}$ and $\bm{c}$.
This gives us the first-order condition and estimator of $\bm{\Sigma}$ as
\begin{equation}
\begin{aligned}
\dfrac{n}{2}\bm{\Sigma}^{-1} &= \dfrac{1}{2} \bm{\Sigma}^{-1} \left( \sum_{i = 1}^n \Es{(\bm{x}_i - \bm{\mu}) (\bm{x}_i - \bm{\mu})^T} \right) \bm{\Sigma}^{-1} \\
\Rightarrow n \bm{\Sigma} \bm{\Sigma}^{-1} \bm{\Sigma} &= \bm{\Sigma} \bm{\Sigma}^{-1} \left( \sum_{i = 1}^n \Es{(\bm{x}_i - \bm{\mu}) (\bm{x}_i - \bm{\mu})^T} \right) \bm{\Sigma}^{-1} \bm{\Sigma} \\
\Rightarrow \hat{\bm{\Sigma}} &= \dfrac{\sum_{i = 1}^n \Es{(\bm{x}_i - \hat{\bm{\mu}}) (\bm{x}_i - \hat{\bm{\mu}})^T}}{n}.
\end{aligned}
\end{equation}
As with $\hat{\sigma}^2_{\epsilon}$ this can be written in terms of the sufficient statistics with
\begin{equation}
\begin{aligned}
\hat{\bm{\Sigma}} &= \dfrac{\sum_{i = 1}^n \Es{\bm{x}_i \bm{x}_i^T}}{n} - \dfrac{\hat{\bm{\mu}} \sum_{i = 1}^n \Es{\bm{x}_i}^T}{n} - \dfrac{\sum_{i = 1}^n \Es{\bm{x}_i} \hat{\bm{\mu}}^T}{n} + \hat{\bm{\mu}}\hat{\bm{\mu}}^T \\
&= \dfrac{\sum_{i = 1}^n \Es{\bm{x}_i \bm{x}_i^T}}{n} - \hat{\bm{\mu}}\hat{\bm{\mu}}^T.
\end{aligned}
\end{equation}

\section{Re-expression of $\mathbb{E}[Y X^a_j X^b_k | \boldmath{X}_O]$ Under MDP 1}
\label{a_mdp1y1}

Denote the index set of missing data except $j$ and $k$ as $A = M \setminus \{j, k\}$.
Then we can re-express $\mathbb{E}[Y X^a_j X^b_k | \bm{X}_O]$ as follows:
\begin{equation}
\begin{aligned}
\mathbb{E}[Y X^a_j X^b_k | \bm{X}_O] &= \int_y \int_{x_j} \int_{x_k}  y x^a_j x^b_k \mathbb{P}(y, x_j, x_k | \bm{x}_O) \, dx_j \, dx_k \, dy \\
&= \int_y \int_{x_j} \int_{x_k}  y x^a_j x^b_k \int_{\bm{x}_A} \mathbb{P}(y, x_j, x_k, \bm{x}_A | \bm{x}_O) \, d\bm{x}_A \, dx_j \, dx_k \, dy \\
&= \int_y \int_{x_j} \int_{x_k} \int_{\bm{x}_A}  y x^a_j x^b_k \mathbb{P}(y | x_j, x_k, \bm{x}_A, \bm{x}_O) \mathbb{P}(x_j, x_k, \bm{x}_A | \bm{x}_O) \, d\bm{x}_A \, dx_j \, dx_k \, dy \\
&= \int_{x_j} \int_{x_k} \int_{\bm{x}_A} \E[Y|\bm{X}] x^a_j x^b_k \mathbb{P}(x_j, x_k, \bm{x}_A | \bm{x}_O) \, d\bm{x}_A \, dx_j \, dx_k \\
&= \int_{\bm{x}_M} \bm{d}(\bm{X})^T \bm{\beta} x^a_j x^b_k \mathbb{P}(\bm{x}_M | \bm{x}_O) \, d \bm{x}_M \\
&=  \mathbb{E} [\bm{d}(\bm{X})^T \bm{\beta} X^a_j X^b_k | \bm{X}_O],
\end{aligned}
\end{equation}
where $\mathbb{P}(x_j, x_k, \bm{x}_A | \bm{x}_O) = \mathbb{P}(\bm{x}_M | \bm{x}_O)$ since $M = A \cup \{j, k\}$.

\section{Derivation of $\Prob(\boldmath{X}_M | Y, \boldmath{X}_O)$} \label{sec:proof_mdp2}

Consider a quadratic form of Equation~\ref{regressionModel}:
\begin{equation}
Y = \beta_0 + \bm{\beta}^T_f \bm{X} + \dfrac{1}{2}\bm{X}^T \bm{\beta}_s \bm{X} + \epsilon,
\end{equation}
where, without the loss of generality, we assume $\bm{X}$ is ordered as $X = \bbm \bm{X}^T_O & \bm{X}^T_M \ebm^T$ and $\bm{\beta}_f$ and $\bm{\beta}_s$ are defined as
\begin{equation}
\begin{aligned}
\bm{\beta}_f &\coloneqq \bbm \bm{\beta}^T_O & \bm{\beta}^T_M \ebm^T \\
\bm{\beta}_s &\coloneqq \bbm \bm{\beta}_{OO} & \bm{\beta}_{OM} \\ \bm{\beta}_{OM} & \bm{\beta}_{MM} \ebm,
\end{aligned}
\end{equation}
which contain the regression coefficients that correspond to the first-order and product terms of $\bm{d}(\bm{X})$, respectively.
Note that the diagonal entries of $\bm{\beta}_s$ are zeroes and under MDP 2 $\bm{\beta}_{MM}$ is a zero matrix.
We can then re-write the error term as a function of $\bm{X}_M$ as follows
\begin{equation}
\begin{aligned}
Y &= \beta_0 + \bm{\beta}^T_f \bm{X} + \dfrac{1}{2}\bm{X}^T \bm{\beta}_s \bm{X} + \epsilon \\
&= \beta_0 + \bm{\beta}^T_O \bm{X}_O + \bm{\beta}^T_M \bm{X}_M + \dfrac{1}{2} \bm{X}^T_O \bm{\beta}_{OO} \bm{X}_O + \bm{X}^T_O \bm{\beta}_{OM} \bm{X}_M + \epsilon \\
&= \beta_0 + \bm{\beta}^T_O \bm{X}_O + \dfrac{1}{2} \bm{X}^T_O \bm{\beta}_{OO} \bm{X}_O + (\bm{\beta}^T_M + \bm{X}^T_O \bm{\beta}_{OM}) \bm{X}_M + \epsilon \\
\Rightarrow \epsilon &= Y - \beta_0 - \bm{\beta}^T_O \bm{X}_O - \dfrac{1}{2} \bm{X}^T_O \bm{\beta}_{OO} \bm{X}_O - (\bm{\beta}^T_M + \bm{X}^T_O \bm{\beta}_{OM}) \bm{X}_M \\
\epsilon &= \alpha_0 - \bm{\alpha_1}^T \bm{X}_M,
\end{aligned}
\end{equation}
which follows by defining
\begin{equation}
\begin{aligned}
\alpha_0 &\coloneqq Y - \beta_0 - \bm{\beta}^T_O \bm{X}_O - \dfrac{1}{2} \bm{X}^T_O \bm{\beta}_{OO} \bm{X}_O \\
\bm{\alpha}_1 &\coloneqq \bm{\beta}_M + \bm{\beta}_{OM}^T \bm{X}_O.
\end{aligned}
\end{equation}

From here, we can proceed with a direct derivation of $\Prob(\bm{X}_M | Y, \bm{X}_O)$ as follows:
\begin{equation}
\begin{aligned}
\Prob(\bm{X}_M | Y, \bm{X}_O) &\propto \Prob(Y, \bm{X}_M, \bm{X}_O) \\
&\propto \Prob(Y | \bm{X}_M, \bm{X}_O) \Prob(\bm{X}_M | \bm{X}_O) \\
&\propto \exp \left[ -\dfrac{1}{2} \sigma^{-2}_{\epsilon} (\alpha_0 - \bm{\alpha}_1^T \bm{X}_M)^2 \right] \exp \left[ -\dfrac{1}{2} (\bm{X}_M - \bm{\mu}_c)^T \bm{\Sigma}^{-1}_c (\bm{X}_M - \bm{\mu}_c) \right] \\
&\propto \exp \left[ -\dfrac{1}{2} \sigma^{-2}_{\epsilon} (\bm{X}_M^T \bm{\alpha}_1 \bm{\alpha}_1^T \bm{X}_M - 2 \alpha_0 \bm{\alpha}_1^T \bm{X}_M ) \right] \exp \left[ -\dfrac{1}{2} ( \bm{X}_M^T \bm{\Sigma}^{-1}_c \bm{X}_M - 2 \bm{\mu}_c^T \bm{\Sigma}^{-1}_c \bm{X}_M ) \right] \\
&= \exp \left[ -\dfrac{1}{2} (\bm{X}_M^T (\bm{\Sigma}^{-1}_c + \sigma^{-2}_{\epsilon} \bm{\alpha}_1 \bm{\alpha}_1^T) \bm{X}_M + 2 (\bm{\mu}_c^T \bm{\Sigma}^{-1}_c + \sigma^{-2}_{\epsilon} \alpha_0 \bm{\alpha}_1^T)\bm{X}_M ) \right] \\
&= \exp \left[ -\dfrac{1}{2} (\bm{X}_M - \bm{b})^T \bm{A} (\bm{X}_M - \bm{b}) - \dfrac{1}{2} \bm{b}^T \bm{A} \bm{b} \right],
\end{aligned}
\end{equation}
where we completed the square by using the shorthand notation:
\begin{equation} \label{eq:mdp2_shorthand}
\begin{aligned}
\bm{A} &\coloneqq \bm{\Sigma}^{-1}_c + \sigma^{-2}_{\epsilon} \bm{\alpha}_1 \bm{\alpha}_1^T \\
\bm{b} &\coloneqq \bm{A}^{-1} (\bm{\Sigma}^{-1}_c \bm{\mu}_c + \sigma^{-2}_{\epsilon} \alpha_0 \bm{\alpha}_1).
\end{aligned}
\end{equation}
This ultimately implies
\begin{equation}
\begin{aligned}
\Prob(\bm{X}_M | Y, \bm{X}_O) &\propto \exp \left[ -\dfrac{1}{2} (\bm{X}_M - \bm{b})^T \bm{A} (\bm{X}_M - \bm{b}) \right] \\
\Rightarrow \Prob(\bm{X}_M | Y, \bm{X}_O) &\sim \mathcal{N}(\bm{b}, \bm{A}^{-1}),
\end{aligned}
\end{equation}
where $\bm{A}$ and $\bm{b}$ are functions of $\bm{\theta}$ as defined in Equation~\ref{eq:mdp2_shorthand}.
We note that this is the multivariate generalization to the one-dimensional result found by \citeA{Kim2015}.

\section{Generating $\sigma^2_{\epsilon}$} \label{sec:gen_sig_e}

Given a sampled $\bm{x}$, $\bm{\beta}$, and $R^2_a$, we algebraically solve for $\sigma^2_{\epsilon}$ in the following manner.
First, note that $\sigma^2_{\epsilon}$ can be written as a function of $R^2$ by the identity:
\begin{equation}
\begin{aligned}
R^2 &= \dfrac{\sigma^2_{\hat{Y}}}{\sigma^2_{\hat{Y}} + \sigma^2_{\epsilon}} \\
\sigma^2_{\hat{Y}} R^2 + \sigma^2_{\epsilon} R^2 &= \sigma^2_{\hat{Y}} \\
\sigma^2_{\epsilon} R^2 &= \sigma^2_{\hat{Y}} - \sigma^2_{\hat{Y}} R^2 \\
\sigma^2_{\epsilon} R^2 &= \sigma^2_{\hat{Y}}(1 - R^2) \\
\sigma^2_{\epsilon} &= \dfrac{\sigma^2_{\hat{Y}}(1 - R^2)}{R^2}.
\end{aligned}
\end{equation}
In turn, $\hat{\sigma}^2_{\hat{Y}}$ can be estimated from $\bm{x}$ and $\bm{\beta}$ by the following calculations:
\begin{equation}
\begin{aligned}
\hat{Y}_i &= \bm{\beta}^T \bm{d}(\bm{x}_i) \\
\bm{\mu}_{\hat{Y}} &= \Sum \dfrac{\hat{Y}_i}{n} \\
\hat{\sigma}^2_{\hat{Y}} &= \Sum \dfrac{(\hat{Y}_i - \bm{\mu}_{\hat{Y}})^2}{n - 1},
\end{aligned}
\end{equation}
Finally, the conversion between $R^2$ and $R^2_a$ was obtained by simply algebraically manipulating the definition $R^2_a$
\begin{equation}
\begin{aligned}
R^2_a &= 1 - (1 - R^2) \dfrac{n - 1}{n - d} \\
\Rightarrow R^2 &= 1 - (1 - R^2_a) \dfrac{n - d}{n - 1}.
\end{aligned}
\end{equation}
Thus, using $\bm{x}$, $\bm{\beta}$, and $R^2_a$, we can generate a $\sigma^2_{\epsilon}$ parameter and simulate the errors as $\epsilon \sim \mathcal{N}(0, \sigma^2_{\epsilon})$ to finally obtain outcome with $Y = \bm{d}(\bm{X})^T \bm{\beta} + \epsilon$.

\section{Generating \boldmath{$R$}} \label{sec:gen_mis}

The first step in generating $\bm{R}$ was assigning missing data patterns.
Recall that the simulation parameter $\varphi_{\text{MIS}}$ was the proportion of missing data overall and $\varphi_{\text{MDP3}}$ was the proportion \textit{among the missing data} assigned to MDP 3.
Let us define $\varphi_j$ as the proportion of the overall data assigned to MDP $j$.
We set these quantities as follows:
\begin{equation}
\begin{aligned}
\varphi_1 &= \varphi_2 \coloneqq \dfrac{1}{2} \varphi_{\text{MIS}} (1 - \varphi_{\text{MDP3}}) \\
\varphi_3 &\coloneqq \varphi_{\text{MIS}} \varphi_{\text{MDP3}}
\end{aligned}
\end{equation}
That is, $\varphi_3$ is simply determined by $\varphi_{\text{MIS}}$ and $\varphi_{\text{MDP3}}$, and $\varphi_1$ and $\varphi_2$ share the remaining proportion of missing data is equally.
Subsequently, let us define $n_j$ as the number of samples assigned to MDP $j$.
This was calculated as $n_j = \lceil n\varphi_j \rceil$, or $n\varphi_j$ rounded up.
Therefore, even if $n$ is small, each MDP is assigned at least one case so long as its corresponding $\varphi_j > 0$.

Once all $n_j$ were determined, we randomly (with uniform probability) selected a non-product variable in $\bm{X}$ to serve as a ``missingness anchor'' variable (denoted $X_a$, with mean $\mu_{a}$ and variance $\sigma^2_{a}$) to determine the missingness in all other variables.
This was done through an intermediate latent propensity variable, denoted $R^*$, defined as
\begin{equation} \label{eq:R_star}
R^* \coloneqq \zeta \left(\dfrac{X_a - \mu_{a}}{\sigma_a}\right) + \left( \sqrt{1 - \zeta^2} \right) \epsilon_0,
\end{equation}
where $\epsilon_0 \sim \mathcal{N}(0, 1)$.
This parameterization sets $\text{Cor}(R^*, X_a) = \zeta$, which allows for simple control over the MAR mechanism.
We set $\zeta = 0.7$ for all conditions.
Thus given $n$ samples of $X_a$, we used Equation~\ref{eq:R_star} to draw $n$ samples of $R^*$.
Then to determine missingness, the greatest $n_{MIS} \coloneqq n_1 + n_2 + n_3$ values of the sample of $R^*$ were assigned randomly to MDPs~1, 2, or~3 with uniform probability, such that the total of each MDP was equal to its corresponding $n_j$.

Once MDPs were assigned, missingness was drawn as follows.
First, to guarantee that the assigned missingness pattern was followed, one $X_j$ (or a product-term pair of $X_j$ and $X_k$) had their $R_{X_j}$ set to zero, selected with uniform probability among all MDP appropriate $\bm{X}$ variables.
Then we set
\begin{itemize}
  \item MDP 1: $R_Y = 0$ and $R_{X_j} \sim \text{Bernoulli}(0.5)$, for all $j \neq a$.
  \item MDP 2: $R_Y = 1$ and $R_{X_j} \sim \text{Bernoulli}(0.5)$, for all $j \neq a$ such that $X_j$ is not a factor of any product-terms in $\bm{d}(\bm{X})$. Then for pairs all $(X_j, X_k)$ that form a product-terms in $\bm{d}(\bm{X})$, we draw $R_{X_j}$ sequentially:
  \begin{itemize}
    \item If both $R_{X_j}$ and $R_{X_k}$ have not been drawn yet, then select one with uniform probability and draw as $\text{Bernoulli}(0.5)$.
    \item If only one $R_{X_j}$ has been set, then $R_{X_k} = 1 - R_{X_j}$.
  \end{itemize}
  \item MDP 3: $R_Y = 1$ and $R_j = R_k \sim \text{Bernoulli}(0.5)$, for all pairs $(j \neq a, k \neq a)$ such that $(X_j, X_k)$ form a product-term in $\bm{d}(\bm{X})$.
\end{itemize}

\end{document}